\documentclass[11pt]{article}
\usepackage[margin=1.15in]{geometry}
\usepackage{amsmath,amssymb,amsthm}
\usepackage{booktabs}
\usepackage[round]{natbib}
\usepackage[hidelinks]{hyperref}
\makeatletter
\renewcommand\paragraph{\@startsection{paragraph}{4}{\z@}%
  {3.25ex \@plus1ex \@minus.2ex}{-1em}%
  {\normalfont\normalsize\itshape}}
\makeatother
\newtheorem{theorem}{Theorem}
\newtheorem{proposition}{Proposition}
\newtheorem{remark}{Remark}
\newcommand{\R}{\mathbb{R}}
\newcommand{\E}{\mathbb{E}}

\title{Scalable Inversion of Contests with Correlated Performances,\\
Including Softmax and Multinomial Probit}
\author{Peter Cotton\thanks{\texttt{peter.cotton@microprediction.com}. The reference
implementation (Python), its parity-locked ports (R, Rust, JavaScript) and the seeded
scripts behind every number in this paper live at
\texttt{github.com/microprediction/winning}; this revision's tables were
produced at repository tag \texttt{paper-r11}, package version
\texttt{1.4.0}.}}
\date{August 2026}

\begin{document}
\maketitle

\begin{abstract}
Multinomial probit choice probabilities over $n$ alternatives are
Gaussian orthant integrals, computed by simulation for thirty years, one
expensive integral per alternative. Inversion, which is to say determining item
attractiveness consistent with a prescribed choice probability vector, is even more difficult and has been
considered impractical for correlated contests when $n$ is large. Yet here, for families lying within a grammar including factor, block and hierarchical covariance structures, we exhibit a calibration tested at $n=1{,}000{,}000$ reproducing probabilities to very high accuracy, even in the extreme tail. We must return to much smaller problems for any performance comparison to be possible due to limitations of the prior art. The Geweke--Hajivassiliou--Keane simulator is the standard (and still appropriate for high rank) but is two hundred times slower
already at $n=200$, and its measured cost grows roughly as $n^{2.8}$ while
ours is linear. Furthermore our approach applies to any continuous performance distributions within reason: the Thurstone--Mosteller model families
thereby become a practical alternative to logit at modern scale.
\end{abstract}

\section{Introduction}\label{sec:intro}

In a classic setup cutting across many fields, a decision maker faces $n$ alternatives. Her choice is modeled as a stochastic contest, where each item has some inherent attractiveness (utility) to which noise, often Gaussian, is added to arrive at a random score. The probability of her choosing an item is synonymous with the probability that said item will have the largest score.

The races we consider in this paper come with an additional complexity because the random component for one item is assumed to be correlated with that of another. We are asserting a covariance on performances within a family to be described. It is natural that choice can be influenced by variables common to some or all participants. Sprinters in a bike race might prefer a slower pace. Vacation choices might co-vary with exchange rates.

In psychology an underlying model for choice is required for the computation of counterfactuals: such as the probability that a child will choose waffles if their favorite option, pancakes, were to be removed from the breakfast menu. Increasingly, vision and text machine learning applications present new kinds of selection and routing conundrums, many of which demand choosing from very large collections, such as picking the best model checkpoint from thousands. Indeed, it could be said that the generation of a token is a race, and by that measure, perhaps billions of contests are being held in the time it takes to read this paragraph.

Nature holds races at all scales, such as activation pathways in the brain, reaction sites or lightning strikes. Simulations of physical systems, such as material breaking points, can also be approximated as correlated races. There is a close connection between winning probabilities and local times in Atlas models (ranks of stochastic processes), for another day.

We initially consider the multinomial probit model (an example of the race we speak of) as a point of reference to prior art, even though our method does not impose an assumption of normal performance, nor any particular distribution family for that matter. Write
\begin{equation}
U_i = \mu_i + \varepsilon_i, \qquad \varepsilon\sim N(0,\Sigma);
\end{equation}
the probability that alternative $i$ is chosen is
\begin{equation}\label{eq:orthant}
p_i \;=\; \Pr\{U_i \ge U_j \ \text{for all } j\}
      \;=\; \Pr\{\tilde U \ge 0\},
\qquad \tilde U_j = U_i - U_j,
\end{equation}
an $(n-1)$-dimensional Gaussian orthant integral. The appeal of this inherently natural generative model has been understood since \citet{thurstone1927}. The $\Sigma$ allows alternatives to be close substitutes, but even for $\Sigma=I$ this departs from the choice axiom of \citet{luce1959} when treated as a model for preference revealed by item removal. The axiom would mandate a simple linear renormalization of probability. But if \eqref{eq:orthant} is calibrated to known $p_i$ and then one runner is removed, the new race speaks otherwise except for very special situations (such as a symmetric contest).

Although it is out of scope, most readers would not be surprised if the disagreement between counterfactuals implied by \eqref{eq:orthant} in this manner and Luce's Choice Axiom is often in the right direction for a diverse list of applications \citep{cottonsurvey}.\footnote{The practitioners might be worth listening to on the topic of Luce's Choice Axiom and its possibly dubious linear renormalization of probability. \citet{benter1994},
whose Hong Kong betting operations, by his own concession, made close to a
billion dollars overall \citep{chellel2018}, wrote of the simple
renormalization used to run probabilities down the finishing order (the
formula of \citet{harville1973}, which is the Luce axiom applied to place
and show)
that it ``is significantly biased and should not be used for betting
purposes.''} Practitioners in signal processing and demand estimation have had their own reasons for preferring Gaussian races. Often, this form is independently motivated. It may be demanded by closure. It is consistent with continuous time models and plays better with linear gaussian networks and rating systems. It may simply make more sense.

It is possible to reverse into generative race models that satisfy Luce's Choice Axiom. For instance, we might assume runners have exponential performances (not offset, but with different hazard rates), or we can choose Gumbel-noise variants of \eqref{eq:orthant}, or any monotonic transformation of the same. This is often done unwittingly, however, when simple linear renormalization of probability is applied. It is desirable to develop numerical methods that support all possibilities to remove the evidenced temptation to select a model based on its analytical tractability.

An important part of providing that utility is the share inversion: a partial calibration (fixing the covariance). As a finer point, the inversion is a map from a probability vector back to an {\em equivalence class} of races that yield the same probabilities. The raw covariance is over-parameterized, and there is an even more obvious translation symmetry since only relative performance matters.

\subsection{A key integral}

Unfortunately, the computational curse has been understood since Thurstone himself,
and it already bites in his simplest case. In the original psychometrics
terminology, \eqref{eq:orthant} with equal discriminal dispersions and
zero correlations is Thurstone's Case V. Later \citet{mosteller1951} showed
that the correlations may be relaxed from zero to any common value without
changing the method, and supplied the estimation and testing apparatus,
which is why the family carries both names.\footnote{Mosteller's
relaxation is an early sighting of the identification this paper relies
on throughout. With equal dispersions and a common correlation,
$\Sigma=\sigma^2[(1-\rho)I+\rho\mathbf{1}\mathbf{1}^{\top}]$, and since
$P\mathbf{1}=0$ the common correlation component --- that component, not
the whole covariance --- lies in the choice-irrelevant direction:
$P\Sigma P=\sigma^2(1-\rho)P$, observationally equivalent to independent
noise at variance $\sigma^2(1-\rho)$. The two price every contest
identically, which the engine reproduces to $6\times10^{-17}$. What this
paper adds is the correlation that survives the projection.} The integral \eqref{eq:orthant} has no closed form, and estimation
needs it and its derivatives at every trial value of the parameters. It is no surprise that the considerably easier logit models tend to be far more popular than probit, not only for calculations such as their obvious extension to permutation probability, but also in the very architecture of most of the world's neural networks.

The workhorse answer to the challenge posed by \eqref{eq:orthant} and calibration in particular (for which raw Monte Carlo is clearly problematic) is the GHK simulator of \citet{geweke1989},
\citet{borschsupan1993}, \citet{keane1994} and \citet{hajivassiliou1996}. For
the focal alternative $i$, GHK Cholesky-factors the covariance of the
$(n-1)$ utility differences $U_i-U_j$ and writes the orthant probability
as a nested sequence of one-dimensional truncated normal probabilities,
each conditioning on draws for the previous coordinates --- the very
recursion \citet{domencich1975} wrote down deterministically (their
eq.~4.32) and set aside as too cumbersome for estimation.

GHK's contribution was to simulate it. Averaging $R$ such importance weights gives an
unbiased, smooth estimate of one $p_i$ at cost $O(n^2)$ per draw. The estimator
made maximum simulated likelihood practical and it remains the default in
textbooks \citep{train2009}. Its costs are also well appreciated. Error shrinks as
$O(R^{-1/2})$ and the log of a simulated probability is biased at $O(1/R)$. The
gradients inherit the simulation noise too and a full probability vector costs
$n$ separate runs. A single likelihood contribution needs only the chosen
alternative's probability, so the $n$-fold cost falls not on the individual
log-likelihood but on everything else estimation touches: share inversion,
aggregate-share and market models, prediction and welfare across the menu,
and the derivative of any of these with respect to all $n$ utilities.

This paper takes a different route to the same object, following \citet{cotton2021} initially. Also in keeping with that paper, we will adopt once and for all a min-wins
convention. Performances are to be considered golf scores, or running times where the smallest is the winner. To avoid confusion with utility, we write
$$X_i=-U_i$$
and we consider the probability that any given $X_i$ takes the lowest value.

Our numerical approach requires, but does not stop with, an old idea: conditioning on the winning performance. Indeed \citet{domencich1975} write the general random-utility choice probability as a one-dimensional integral over the winning utility
(their eq.~4.7, restated as eq.~4.26 although they did not take the step to turn that into a method as with \citet{cotton2021}). Without special structure, the integrand conceals a multivariate calculation, only the Weibull case yields a convenient form, and the multiple-choice probit generalization
was, in their words, ``computationally intractable.''

Under \emph{independence} the integrand remarked on in \citet{domencich1975} really is a density times a product of univariate distribution functions, and they price exactly that form for independent Cauchy shocks (their eq.~4.38) one alternative at a
time, concluding it is obtainable by ``straightforward, but costly''
numerical integration. The sequential-conditioning
decomposition that GHK later simulates (their eqs.~4.27--4.32, set aside
as too cumbersome for iterative estimation) was present in 1975,
and judged impractical.

The expression of winning probability for the case of independent performances is standard in competing-risks
calculus: $G=\prod_k S_k$ is the all-causes survival and $f_i/S_i$ the
cause-specific hazard, so $p_i=\int G\,(f_i/S_i)$ is the classical
cause-probability formula \citep{kalbfleisch2002}.
\citet{fosgerau2013} considering random-utility models treat the expected maximum as potential  and its gradient as probabilities. Choice probability is
an integral along the common maximum level.

Our starting point is similarly this well-beaten integral:
\begin{equation}\label{eq:onedim}
p_i \;=\; \int f_i(x)\,
\Pr\{X_j>x\ \text{for all } j\ne i\mid X_i=x\}\,dx,
\end{equation}
where $f_i$ is the density of $X_i$. The simple interpretation is that we integrate over the winning time $x$,
requiring every rival to come in slower. When performances are
independent, the conditional probability factorizes into a product of
marginal survivals.
\begin{equation}\label{eq:field-intro}
p_i \;=\; \int f_i(x)\prod_{j\ne i}S_j(x)\,dx,
\end{equation}
with $S_j$ the survival of $X_j$. Under the covariance grammars developed
below the same product holds conditional on a few shared variables and is
then averaged over them, so the winning-value integral stays
one-dimensional and correlation adds only a low-dimensional outer
integral.

\subsection{Contribution}

For independent performances \citet{cotton2021} computed \eqref{eq:field-intro} for
all $n$ on a lattice in a novel manner by building the field $\sum_j \log S_j$ once and dividing
each alternative out, and inverted the map from probabilities back to $\mu$. Because performances were atomic distributions represented on a lattice a ``multiplicity calculus'' that tracked the expected number of ties conditioned on the winning score was required. That method is robust to reasonably pathological choices of performance distributions defined on a lattice (lumpy, gaps and so on) so it is preserved in our software.

The present approach trades some of that flexibility for an implicit assumption that performances are smoothly varying (in the loose sense) in order to lean more heavily on quadrature and achieve, empirically, much tighter inversion. The
present paper also extends the field construction to correlated performances by
conditioning. Whenever $\Sigma$ makes performances independent given a few shared variables, the field factorizes and all $n$ probabilities cost $O(nLQ)$ for a lattice of $L$ points and $Q$ quadrature nodes. The conditioning technique per se is standard, needless to say.

It tends to be the case in financial practice and elsewhere that structured covariance is assumed. We do not claim a general solution for arbitrary covariance matrices, but nor do these tend to be advocated \citep{train2009}. Our intent is that the grammars here cover the structured families in typical use. The first six rows of Table~\ref{tab:grammar} are represented without covariance fitting.

Within such a class of covariance matrices (we will term it the factor grammar) the derivatives are computed in the same pass at the same cost. This yields Jacobian--vector products and own-survival curves. The cost is also $O(nLQ)$, where again $L$ dictates the lattice size and $Q$ the number of quadrature points.

We will show that the Jacobian is a weighted graph Laplacian. It is applied as an operator and is never materialized in the code. The block and nested kernels materialize their Jacobians, and the tree's cross-cluster block is approximate. The most important distinction from the prior art is that there is no need to run the entire algorithm for each choice of $i$, a point we will return to shortly when discussing the existing methods.

The error is controlled by quadrature. It is deterministic for the Gauss--Hermite rules and reproducibly randomized for
scrambled-Sobol nodes.

Our caveats aside regarding the general covariance problem, we provide in the accompanying code a method for fitting arbitrary covariance to the grammar we specify, as developed in \S\ref{sec:general}. However, for truly high rank problems where accuracy is paramount, we make no claim of superiority over GHK; on the contrary, it should remain the appropriate choice if the problem at hand fails to establish otherwise.

\begin{table}[h]\centering\small
\begin{tabular}{lll}
\toprule
model & covariance grammar & source \\
\midrule
Luce / softmax & independent, minimum-Gumbel base, $D=\pi^2/6$ & exact identity \\
Thurstone--Mosteller \citep{mosteller1951} & independent, normal base & \citet{cotton2021} \\
factor multinomial probit & $VV'+\mathrm{diag}(D)$ & \S\ref{sec:factor} \\
teams, conferences & blocks: rank-$r$ effect per cluster & \S\ref{sec:blocks} \\
market plus sectors & nested: blocks $\times$ global factor & \S\ref{sec:blocks} \\
hierarchies & tree: uniform shared effects & \S\ref{sec:tree} \\
dense $\Sigma$ & fitted grammar $+$ residual factors (approx.) & \S\ref{sec:general} \\
\bottomrule
\end{tabular}
\caption{The covariance grammar.}\label{tab:grammar}
\end{table}

Our chief contribution is scaling. We seek to bring probit modeling to a qualitatively different problem space.  To put it in perspective, we note that on one laptop, the forward computation of all-$n$ probabilities from known utilities, assuming a rank-one factor and using the Rust kernels, takes $0.18$, $2.7$, and $29$ seconds respectively for $n=10^{4},10^{5},10^{6}$. That speed was unthinkable previously.

The inversion is the tougher problem yet at $n=10^{6}$ this takes not much more than a minute for the independent case, and the rank-one factor problem can be solved at this
extreme scale also, in around $20$ minutes.

Comparisons die well before this scale, as far as we know. Other approaches are considered in detail presently, but briefly, if the comparison set is taken to be those approaches providing smoothness suitable for inversion, GHK is generally regarded as leading that category.

It cannot compete. We estimate a two order of magnitude speed gap at $n=200$ that is steadily worsening with $n$. GHK's empirical complete-vector cost is superquadratic and steepening as we pass beyond that scale.\footnote{Empirically, the log-log slope is
$2.15$ from $n=10$ to $200$ and $2.8$ from $200$ to $1000$. The times at
$R=10^{3}$ were $0.57$, $5.66$, and $51.3$ seconds at $n=200,500,1000$,
scaling by ten at $R=10^{4}$.} In fact in our tests, the gap to our method grows almost quadratically with $n$ beyond $n=200$ because in that regime GHK scales at roughly $n^{2.8}$ while our shared-field pass is approximately linear in $n$ at fixed
lattice and node budgets. We have provided in Appendix~\ref{app:extrap} some brave extrapolations of GHK performance to
$10^{4}$--$10^{6}$ alternatives and the wall clock time ratios, so dramatic as to not be reproducible in any reasonable time.\footnote{See \texttt{bench.py}.}

We proceed as follows. Section \ref{sec:alternatives} places the method among the alternatives.
Section \ref{sec:lattice} gives the independent lattice race, which is the same problem considered in \citet{cotton2021} solved slightly differently, as noted. Then Section
\ref{sec:grammar} extends this across the covariance grammar. The grammar allows the set of covariance matrices to include block and tree structures with global and local factors, each one written out as a recursion. It may be of interest to those in the financial community that this structure includes that assumed in Hierarchical Risk Parity.\footnote{Although it is beyond the present scope, we mention that markets are also an example of choice probability vectors, in the sense that a dollar portfolio is allocated across stocks. The present work was motivated in part by a desire to calibrate to the market share, or to a given portfolio, under a known covariance in order to be able to re-run the same ``race'' with different assumptions or universe.}

Section \ref{sec:jac} carries the key analytical insight and presents
the Jacobian operator and its inversion. Section \ref{sec:general} reduces
dense covariance to the grammar family, approximately, and reports accuracy
per named random-correlation ensemble. Section \ref{sec:num} contains the
benchmarks.

\section{The alternatives}\label{sec:alternatives}

\paragraph{Deterministic quadrature.} \citet{genz1992} transformed the orthant
integral to the unit cube and integrated by adaptive and quasi-Monte Carlo rules.
This is somewhat of a standard in \texttt{R} because \texttt{mvtnorm} chooses this path, presumably for accuracy. It certainly is accurate. The benchmark below confirms that: its answers agree with our lattice to within the sampling noise of the simulation that both are checked against. However, the disadvantage lies in what a single call returns. Genz's method evaluates one orthant probability only, which is the win probability of one alternative. Pricing a contest
requires all $n$ of them, so the routine must be called $n$ times, and each of
those calls is itself an $(n-1)$-dimensional integral over the focal
alternative's difference orthant. At $n=30$ the complete vector therefore
costs about three hundred times what the lattice costs, and since the work per
call also grows with $n$, that ratio widens as the field grows.

\paragraph{Frequency simulation.} The crude frequency estimator of
\citet{lerman1981} counts argmax draws. It is unbiased for the whole vector at
once but not smooth in parameters, and its per-entry relative error explodes for
small $p_i$. Measured below: at wall clock matched to the lattice it carries
twelve to twenty times the total-variation error, records zero counts on tail
alternatives at $n=30$, and at ten times the budget remains a factor of two to
four behind.

\paragraph{GHK.} As described above. It is smooth in parameters but like Genz it prices one alternative per run. Despite this, GHK is seemingly the most widespread. Stata's \texttt{cmmprobit} manual states that the likelihood evaluator ``implements the Geweke--Hajivassiliou--Keane (GHK) algorithm to approximate the multivariate distribution function''
\citep{statacorp2023,gates2006}, and the command's entire
\texttt{intmethod} menu (Hammersley, Halton, or pseudorandom point
sets) selects only the sequence feeding the same simulator. The
market-leading implementation offers no non-GHK method at all. R's
\texttt{mlogit} and \texttt{mvProbit} also default to GHK
\citep{croissant2020}, and the GHK-based \texttt{mvprobit} of
\citet{cappellari2003} carried the method into applied Stata work at
large.

\paragraph{Stern.}
\citet{stern1992} is the closest ancestor of the construction used here.
He split the error into an independent normal component and a correlated
remainder, observed that conditional on the remainder the alternatives
are independent so that the choice probability is an analytic product of
univariate normal c.d.f.s, and averaged that analytic expression over
Monte Carlo draws of the remainder. The motivation was also smoothness, and
similarly, Stern is testimony to a common willingness to accept a substantial per-draw cost to buy a simulator that could be differentiated. Our work differs because the conditioning variables here are a
low-dimensional factor structure chosen so that the outer integral is
quadrature rather than simulation: deterministic Gauss--Hermite by
default, and seeded scrambled Sobol when the integrand is sharp, which
replaces Monte Carlo noise with reproducible randomized-quadrature error
rather than eliminating integration error outright. Either way there is
no per-draw sampling noise in the returned probabilities, and, as we will discuss, the conditional problem is solved on a lattice that returns all $n$ alternatives from a single sweep rather than one probability per run, with derivatives coming out of that same
sweep.

\paragraph{Distributionally robust choice.} A different escape from the
orthant integral fixes only the first two moments of the utilities and
prices the \emph{extremal} joint distribution: the cross-moment model of
\citet{mishra2012} computes, by semidefinite optimization, the choice
probabilities under the moment-feasible distribution that maximizes
expected agent utility, and \citet{ahipasaoglu2019} reformulate that
computation as a convex program that scales to many alternatives. This
answers a genuinely different question from ours --- the worst case over
a moment class rather than a specified law, with representative-agent
rather than distributional semantics --- so the two are complements: that
family avoids ever specifying the joint distribution, while this paper
prices and inverts the distribution the modeller actually specifies.

\paragraph{Bayesian data augmentation.} \citet{albert1993} and
\citet{mcculloch1994} sidestep the integral entirely by Gibbs sampling the latent
utilities; \citet{imai2005} implement the marginal data-augmentation version
in the \texttt{MNP} package, and \citet{loaizamaya2022} carry the Bayesian
route to large choice sets by variational methods. For Bayesian inference
this is a complete and often excellent answer. It is simply answering a different question than the
one benchmarked here: the cost moves into the Markov chain, the likelihood
surface is never formed, and choice probabilities appear only as posterior
functionals, so there is no pricing call to time or to check.

\paragraph{Analytic approximations.} \citet{clark1961} matched moments of
pairwise maxima; \citet{mendell1974} and \citet{solow1990} approximate by
sequential conditioning on first and second moments. These are fast, sometimes
startlingly accurate, and carry no error control, and the benchmark below
measures all three properties at once: Mendell--Elston prices the $n=10$ vector
in two milliseconds within $2\times10^{-3}$ of the reference, then drifts to
$3\times10^{-2}$ at $n=30$ with tail probabilities off by a factor of nearly
four, and nothing in the output announces the change. The Clark-type recursion
behaves alike. The family is current practice, not history:
\citet{bhat2011} builds the MACML estimator on precisely these
approximations, and the \texttt{Rprobit} package implements it
\citep{bauerRprobit}, with its asymptotics examined by
\citet{batram2019}. Expectation propagation
is the modern member of the same family. We benchmark representative
sequential moment approximations through the classical members; MACML's
composite-likelihood construction around its analytic ingredient is its own
estimator and is not itself benchmarked here.

\paragraph{Minimax tilting.} \citet{botev2017} computes single Gaussian
rectangle probabilities with bounded relative error deep into the tails,
in dimensions up to roughly a thousand. It is the accuracy reference, and
we use it as the referee for our own tail claims. Like GHK it prices one
alternative per run, so a whole race costs $n$ runs, and which method is
preferable depends on the question. For a single tail probability Botev's
is the right tool and this paper defers to it. For the all-$n$ vector it
is both slower and less accurate than the lattice at every size measured:
at $n=50$, $0.25$ s against $0.07$ s and $2.5\times10^{-3}$ error against
$2.9\times10^{-5}$; at $n=200$, $25.8$ s against $0.29$ s and
$1.4\times10^{-3}$ against $4.1\times10^{-5}$. Raising the tilting sample
to $R=10^{4}$ narrows the accuracy gap and widens the time one: at
$n=200$ that is $37.9$ s, a hundred and thirty times the lattice's cost
and still an order of magnitude short of its error.

\medskip

The reader will note the common thread. Smoothness has been, and continues to be, bought at a significant performance cost. The cost, compared to Monte Carlo, is that one alternative is computed at a time, and this cost seems to have been accepted as unavoidable. We will show that it is not, at least within the covariance grammar we specify.

We offer a different Faustian bargain that will often, but not always, be preferable. In exchange for the assumed covariance structure, which is very often assumed anyway, the user does not have to choose between
speed, accuracy, scale, or smoothness. They can have all four, and very often they need them for downstream modeling tasks beyond share inversion, such as welfare analysis, game theory, and propagation of derivatives to other parts of a system (for ratings or belief networks).

\section{The lattice race}\label{sec:lattice}

The application programming interface for \citet{cotton2021} asked for numerical representation of performance densities on an equi-spaced lattice. The current method instead requires this assumption in the form of user-supplied functions.\footnote{This is a crucial distinction. The lattice in the current work is only a computational device. The continuous race is more directly considered, compared to \citet{cotton2021} where the discrete race with explicit ties was the object under study and its relationship to the continuous race assumed.}

Emphasizing again our lowest-score-wins convention (arguably somewhat unfortunate for the special case of probit), let $X_i=\mu_i+\sigma_i
\epsilon_i$ with $\epsilon_i$ i.i.d.\ from a standardized base density (normal,
Gumbel, or any density supplied as code), and write $f_i$, $S_i$ for the density
and survival of $X_i$. From here on, $\mu_i$ is to be considered a \emph{running time}
location, the negative of the introduction's utility location. For avoidance of any doubt, raising
it makes contestant $i$ worse and lowers $p_i$. Inversion returns time
locations, and utility coefficients are their negation.

The algorithm evaluates \eqref{eq:field-intro} on a lattice $x_1<\dots<x_L$:
\begin{enumerate}\itemsep2pt
\item compute $\log S_j(x_\ell)$ and $\log f_j(x_\ell)$ for all $j,\ell$;
\item accumulate the field $L(x_\ell)=\sum_j \log S_j(x_\ell)$;
\item for each $i$, form the integrand
      $f_i(x_\ell)\,e^{L(x_\ell)-\log S_i(x_\ell)}$, which is contestant $i$'s
      density times everyone else's survival, with $i$ divided out of the shared
      field in log space;
\item sum times the lattice spacing, and normalize the vector.
\end{enumerate}
The cost is $O(nL)$: the field is built once, not once per contestant.

\subsection{Details of an adaptive lattice choice}

We write $$L(x)=\sum_j \log S_j(x)$$ for
the log survival field, evaluated once per lattice point across all
contestants, each contestant's leave-one-out product is a subtraction
rather than a division: $$\prod_{j\neq i}S_j = \exp(L-\log S_i).$$ This is
what makes one sweep price the whole field, since the alternative is to
re-form an $(n-1)$-fold product for each of $n$ contestants. It is also
what keeps the far tail finite. There, every $S_j$ is small enough that
$\prod_j S_j$ underflows to zero in double precision, and the division
form returns $0/0$ for any contestant whose own survival has underflowed;
a single such point makes the integral undefined. As a sum of logarithms
$L$ is an ordinary negative number of order $-10^{3}$, and $L-\log S_i$
is a finite difference.\footnote{Survivals are floored at $10^{-300}$ inside every
base so that $\log S$ is never $-\infty$, and the exponent is clipped to
$[-745,0]$: below $-745$ the exponential is zero in double precision
anyway, so a hopeless contestant's contribution decays to zero rather
than failing, and above $0$ the value would exceed one, which no product
of survivals can.}

We are economical with the lattice points. The lattice must cover the winning time, not the spread of abilities. However
hopeless a contestant is, it wins only by finishing in the range where
winning times fall; thus, outside that range, the integrand is negligible. In contrast, a lattice that was stretched across the whole ability span would waste compute as most of its points lie where nothing happens.

There is a key defensive measure taken, however. Let $$G(x)=1-\prod_j S_j(x),$$ the probability that somebody has finished
by $x$, which is the distribution function of the winning time. It is
built from the caller's own base survival (rather than a normal
surrogate, say), and it is made conservative. For
the lower endpoint every contestant is placed at its most favourable conditional location, and for the upper endpoint at its least favourable.

Those componentwise extremes need not occur at any single factor node, which is what makes the result an envelope: wider than any one conditional race requires, and covering every conditional race represented by the retained factor nodes.\footnote{This is akin, in cycling, to assuming a fast pace to floor the worst performance of sprinters, and a slow pace to ceil the best performance of climbers.}

Factor realizations outside the node set belong to the outer quadrature error, not to the lattice
window. Solving $G(x)=\delta$ and
$G(x)=1-\delta$ by bisection gives the two ends of the region the race
is decided in. The lattice is placed there with two standard deviations of padding at
each end for safety. At the achieved $\delta$ (which may be relaxed, as
described next), the truncated tails are controlled separately from the
bulk discretization error, and in the benchmark configurations they lie
below the accuracy floors reported shortly.

We also need a starting interval known to contain the answer. Bisection
halves an interval repeatedly and finds the point where $G$ crosses a
given level only if that point lies inside the interval it began with. A
fixed starting interval does not always contain it. Nine standard
deviations beyond the outermost contestant is wide enough for a Gaussian
race, whose survival decays exponentially, but not for a base with
polynomial tails: with Student-$t$ performance at four degrees of
freedom, the point where $G$ falls to $10^{-12}$ lies $1456$ standard
units out.

Bisecting an interval that excludes the crossing neither fails nor warns.
It converges to whichever endpoint is nearer and returns that instead.
Each endpoint is therefore moved outward, doubling the distance each
time, until $G$ at the lower endpoint is at most $\delta$ and $G$ at the
upper endpoint is at least $1-\delta$. Only then is the interval
bisected.

A finer point is that $\delta$ cannot always be honoured. The
lattice has a fixed number of points, chosen by the caller, spread evenly
across whatever window the quantiles ask for, so a wider window means
coarser spacing. The two error sources therefore move in opposite
directions: shrinking $\delta$ removes truncated tail mass but adds
discretization error in the bulk, where nearly all of the probability is.

An example is instructive. For a Student-$t$ race, $\delta=10^{-12}$ asks for a window some
$1456$ standard units wide, which at the default $257$ points places them
$5.7$ standard units apart. That is far coarser than the scale the
integrand varies on. This coarseness causes much more error than would be sustained by missing the tiny tail.

So $\delta$ is treated only as a request. If the window it asks for cannot be
covered at a spacing of half the smallest performance standard deviation
in the field, $\delta$ is multiplied by one hundred and the window
recomputed, repeating until the spacing is affordable.\footnote{In the
current implementation, the value finally used is reported in a warning.
The window is then the stated construction at a stated $\delta$, and a
caller who wants the original one back can have it by raising the number
of points.}

We determined that it is very important to use the caller's own
survival rather than a normal surrogate. Continuing the Student-$t$
example, a base at $\nu=2.5$ has polynomial tails, and a normal envelope
cuts the window short of them: at $n=40$ that costs
$5\times10^{-3}$ of total variation, against $5\times10^{-5}$ when the
true survival sets the window.

Placing the lattice on the bulk is also what makes small lattices work.
Thirty-three points on this window reach an accuracy that five hundred
points spread across the ability span cannot, because the span window
truncates the winner's tail and its error cannot fall below about
$6\times10^{-11}$ however many points are added to it. We found that on
fields with a large number of extreme longshots, the bulk window is a
hundred times more accurate for the same number of points.

\subsection{Derivatives}

A key output of the pass is the derivative of winning probability with respect to ability:
$$
\partial p_i/\partial\mu_j.
$$
This can be used in several ways. A full Newton method can wield the entire Jacobian, although we do not recommend that for large $n$. The diagonals can be used to precondition. And a single row of the Jacobian can be used for an estimation score. Section~\ref{sec:jac} interprets the matrix; what
matters here is how it is computed.

Prior to that we mention a lesson from the implementation, albeit one
that might strike some readers as obvious: a derivative is accurate only
if it differentiates the map that was \emph{actually evaluated}, rather
than the continuum integral that map approximates. This is not a fine point. The lattice should be the very same one. Both the full Jacobian and the single-row version build their window through the routine the forward
pass uses, so the two integrate over identical points.

Another potential trap lies in the fact that the diagonal should be integrated rather than assumed. Raising $\mu_i$
moves contestant $i$'s own density, contributing $\partial
f_i/\partial\mu_i$ to the integral. In the continuum this need never be
computed, because a common shift moves no probability, so each row sums to zero
and the diagonal follows from the off-diagonal entries. However, on a finite
lattice that identity holds only to quadrature accuracy. For this
reason, we determined that the diagonal should be computed directly.

In a similar vein, the normalization must be differentiated too. The lattice masses do
not sum to one \emph{exactly}. The returned probabilities are $p=a/T$ with
$T=\mathbf{1}^{\top}a$. Differentiating that quotient gives
$(A-p\,\mathbf{1}^{\top}A)/T$, where $A=\partial a/\partial\mu$. The
correction term vanishes in the continuum, where $T$ is constant. It does not vanish on a lattice.

The test is finite differences of the forward map itself: shift $\mu_j$,
re-price the field, and compare the change in $p$ against the computed
column. Write $J$ for the Jacobian built as above, and $J_0$ for the
continuum version, which takes its diagonal from the zero-row-sum
identity and omits the normalization correction.\footnote{On a coarse $65$-point Gaussian lattice $J$ matches
to $2\times10^{-11}$ and $J_0$ to $1\times10^{-8}$; under Student-$t$
performance at four degrees of freedom, to $1\times10^{-6}$ and
$1\times10^{-5}$. On a fine Gaussian lattice, where the quadrature error
is far below either figure, $J$ and $J_0$ cannot be told apart. To be clear these are the largest disagreement with the finite differences, relative to the largest
entry of the matrix.}

Columns and row sums do not behave the same numerically. Columns sum to zero identically. The quotient
divides by the lattice total, so $\mathbf{1}^{\top}J=0$ whatever the
lattice does, at $6\times10^{-17}$ from the coarsest grid to the finest.
Rows sum to zero only as the quadrature converges (translation
invariance is a property of the integral, but not the actual sum as implemented). The row sum discrepancy is a useful diagnostic.\footnote{The deviation is very small. It is
$1\times10^{-3}$ at $25$ points,
$2\times10^{-10}$ at $65$, and machine precision at the $257$-point
default, under either the Gaussian or the Gumbel base.}

\subsection{Benefits of maintaining the survival field}

Two byproducts are cheap once the field is built.

The first is the density of a dead heat, which is to say the chance that $i$ and $j$
finish together, per unit of finishing time. It is one more sum over the
same lattice. Section~\ref{sec:jac} shows that this quantity is exactly
the off-diagonal entry of the derivative matrix $\partial p/\partial\mu$,
so one computation answers both questions.

The second is a removal counterfactual: the probability that $j$ wins
once $i$ is taken out of the field. In the continuum this is one more
division, since deleting $i$ simply drops $S_i$ from the product. However, numerically, to continue our theme, care is required. Suppose the lattice is placed where the
\emph{original} winner finishes, but the survivors of a deletion need not
finish there. So, for example if we remove a dominant favourite, the race is decided where the original runner-up would have finished and it is entirely possible that the winner-bulk lattice would not reach that zone.\footnote{
Take $\mu=(-20,0,1)$ with unit variances, so the first contestant wins
essentially always and the bulk window sits at about $[-29,-11]$. Remove
it and the answer is $\Phi(1/\sqrt2)=0.760$, writing $\Phi$ for the
standard normal distribution function. Dividing the favourite out on the
original grid instead leaves a total lattice mass of $3\times10^{-28}$
where a correct integration would give one, so renormalizing it is
renormalizing rounding error. Removals are therefore integrated on their
own span window, with the total mass checked and a defect raised rather
than normalized away.}

\section{The covariance grammar}\label{sec:grammar}

Now we turn to a fixed set of covariance choices that can be accommodated with extremely high accuracy. Naturally, some structures make contestants conditionally independent given a few shared variables. Without that, we have no algorithm so, in increasing order of hierarchy, we consider:
\begin{enumerate}
\item a global factor,
\item factors per cluster, and
\item a tree of uniform shared effects.
\end{enumerate}
Each keeps the shared
field and the $O(nLQ)$ volumetrics but they come with additional numerical lattice considerations, the topic of this section. Each is exact given its structure;
estimating that structure from data is a separate problem we do not address
here. The approximate reduction of a dense covariance to this family is considered, but deferred to \S\ref{sec:general}.

\subsection{Factor covariance with sharpness as a guide}\label{sec:factor}

Let $X=\mu+Vf+\mathrm{diag}(\sqrt D)\,\epsilon$ with $f\sim N(0,I_k)$
independent of $\epsilon$, so $\Sigma=VV'+\mathrm{diag}(D)$. Conditional on $f$
the performances are independent with means $\mu+Vf$, and
\begin{equation}
p \;=\; \sum_q w_q\, p^{\mathrm{ind}}(\mu+Vf_q),
\end{equation}
a $Q$-node quadrature over runs of the independent algorithm at a cost of
$O(nLQ)$.

We are merely averaging the independent race over the shared luck, so to speak. This device is far from new, of course; indeed
\citet{butler1982} built exactly this quadrature for the one-factor
random-effects probit. The \texttt{lpRR}/\texttt{slpRR} interfaces of
\texttt{mvtnorm} integrate over the factor dimensions of
$BB'+\mathrm{diag}(D)$ by Monte Carlo. As we noted at the outset, each of
those integrals returns only one alternative's probability, which is one
reason their performance dies at scale.

The survival field changes that: one pass shares the survival product across all
$n$ alternatives, and \S\ref{sec:num} measures the difference against the
per-alternative version of this same quadrature. The same pass returns the own-parameter slopes $\partial p_i/\partial\mu_i$
at no extra cost (negative in this min-wins convention: raising a time
mean lowers the win probability), which \S\ref{sec:jac} uses to
precondition the inversion.

However, we must pay attention to numerical issues and lattice choice in particular. The integral over $f$ is $k$-dimensional, and experiments have taught us that the rule that discretizes it
has to depend on $k$ and the shape of the integrand.

Gauss--Hermite is the natural rule for a single Gaussian dimension.
Applying it in $k$ dimensions means taking every combination of its
nodes, which costs $q^{k}$ of them. At the default $q=15$ that is
$3{,}375$ nodes at rank three and $50{,}625$ at rank four, but
$759{,}375$ at rank five. So the product rule is used while it stays
under $10^{5}$ nodes, which in practice means rank four and below.

The other consideration is sharpness. What decides the race is not how
far the factor moves a contestant but how far it moves contestants
\emph{apart}: a factor loading every contestant identically shifts the
whole field and changes nothing, however large it is. The intrinsic
statistic is therefore pairwise,
\[
s_{\Delta}=\max_{i<j}\frac{\lVert V_i-V_j\rVert}{\sqrt{D_i+D_j}},
\]
the shared variation available to a contrast against that contrast's own
noise, but computing it is $O(n^2 k)$. The implementation first
gauge-fixes the loadings, $V\leftarrow PV$ with
$P=I-\tfrac1n\mathbf{1}\mathbf{1}^{\top}$ --- a common loading column
adds the same draw to every contestant and cannot move an argmin, so the
centered matrix prices the identical race, and every downstream decision
becomes invariant under $V\mapsto V+\mathbf{1}c^{\top}$ --- and then
dispatches on the linear-time bound
\[
s_{\mathrm{trigger}}\;=\;\sqrt2\,
\max_i \frac{\lVert (PV)_i\rVert}{\sqrt{D_i}}\;\ge\;s_{\Delta},
\]
which follows from the triangle inequality on centered rows. Notice that the bound cannot miss a pair whose contrast sharpness exceeds the threshold as false positives cost extra nodes, not error.

We also remark that the centered maximum is a
common-shift-invariant proxy for factor-space smoothness, but not intended as a complete
characterization of it. The dispatcher is therefore wise to use the
conservative side of it. But nothing needs to be guessed and everything follows
from $V$ and $D$ before any integration is done.

When the trigger is large the factor draw all but decides some pairwise contest, the
quantity being averaged over $f$ is close to a step function, and a rule
built for smooth integrands converges slowly on it.\footnote{Another lesson: the uncentered
maximum row norm $\max_i\lVert V_i\rVert/\sqrt{D_i}$ is tempting because
it skips the centering. But it obeys no such inequality. Instead centering can
\emph{increase} the largest row norm. A three-runner race with loadings
$(-2.9,0)$, $(2.9,0.01)$, $(2.9,-0.01)$ at unit $D$ has uncentered
maximum $2.90$, under the threshold, while its true pairwise sharpness
is $4.10$; a dispatcher trusting the raw norm stays with Gauss--Hermite
there and ships a total-variation error of $6.1\times10^{-3}$ against
the analytic three-runner orthant probabilities.}

A fixed rule of fifteen nodes per dimension can then miss up to five
percent of the total variation, and the error is easy to overlook because
the usual check does not see it. Adding lattice points does not reduce
it: the lattice discretizes the finishing time, while the loss is in the
average over $f$, so refining $L$ returns a stable answer that is stably
wrong.

We implement two remedies. The node order scales with the trigger rather
than staying fixed, within the rank limit above. And past a trigger of
$3$ at rank two or more the
family changes, scrambled Sobol points replacing the product rule, since
more nodes of a smooth rule is the wrong medicine for a near-step
integrand. Rank one with extreme sharpness gets a third treatment, an
equal-weight grid on quantiles, because Gauss--Hermite clusters its nodes
near the centre and that is the wrong place when the integrand is nearly
a step.

\subsection{Blocks and nested covariance}\label{sec:blocks}

The implementation also supports groups. Assign contestant $i$ to cluster $c(i)$ with loading $v_i\in\R^r$ and let
\begin{equation}
X_i=\mu_i+v_i'a_{c(i)}+\sqrt{D_i}\,\epsilon_i,
\qquad a_c \ \text{i.i.d.}\ N(0,I_r).
\end{equation}
This is a block-diagonal rank-$r$-plus-diagonal covariance. It can be useful for teams, conferences, and sibling
products, for example. And it falls into our framework because, conditional on its own cluster's effect, a contestant is independent of its teammates. Distinct clusters are
independent altogether so the joint survival factorizes across clusters,
\begin{equation}\label{eq:blockfield}
\Pr\{X_j>x\ \text{for every } j\} \;=\; \prod_c G_c(x),
\qquad
G_c(x)\;=\;\E_a \prod_{j\in c} S_j(x\mid a)
       \;=\;\sum_q w_q\, e^{S_c(x,a_q)},
\end{equation}
with $S_c(x,a)=\sum_{j\in c}\log S_j(x\mid a)$ the cluster's conditional
log-survival. To say this another way, $G_c(x)$ is just the probability that every member of
cluster $c$ survives past $x$, with the cluster's shared luck integrated out.

Contestant $i$ then wins against a ``cavity'' field ($i$ removed):
\begin{equation}\label{eq:cavity}
p_i \;=\; \int h_i(x)\prod_{c\ne c(i)}G_c(x)\,dx,
\qquad
h_i(x) \;=\; \sum_q w_q\, f_i(x\mid a_q)\,e^{S_{c(i)}(x,a_q)-\log S_i(x\mid a_q)}.
\end{equation}
The term $h_i$ couples $i$'s win density to its own teammates within the shared
draw of $a_{c(i)}$; the cavity product supplies the rest of the field.\footnote{Here we borrow the ``cavity'' terminology from statistical physics, because it is reminiscent of computing the influence on one site by deleting it,
asking what the remaining system does in the hole left behind, and then
reinserting it. A similar construction appears as the cavity distribution
in belief propagation. Note that here the deletion is quite literal and $\prod_{j\ne i}S_j$
is the field with contestant $i$ removed. Readers from econometrics may
prefer \emph{leave-one-out}. The closest familiar object is the inclusive
value of nested logit, an aggregate over alternatives that enters each
alternative's own probability, with the difference that the focal
alternative is excluded from the aggregate here rather than included.
\emph{Field}, likewise, is physics for $\sum_j\log S_j$, the single
aggregate quantity every contestant is priced against, and the reason one
pass serves all $n$ of them.}
The cost
is $O(nLQ_a)$ whatever the number of clusters, because the per-cluster factors
are accumulated once. Rank two is special: the conditional integrand is smooth
and a tensor Gauss--Hermite rule beats scrambled Sobol at equal nodes by a wide
measured margin.

One limitation of the hierarchical kernels is stated here but we aim to remove it.
Their quadrature is a fixed-order rule, without
the sharpness heuristic of \S\ref{sec:factor} being applied. So when a cluster loading
is large against the idiosyncratic noise the conditional integrand is a
near-step and Gauss--Hermite fails at any order (measured at cluster
sharpness $18$: $5\times10^{-2}$ of total variation against a
four-million-path referee).\footnote{At time of writing, the shipped kernels warn past sharpness $3$ and
direct such fields to the factor grammar whose node family escalates
automatically. However porting that escalation to the hierarchical kernels is not yet complete.}

The hierarchical kernels place their lattice by the same envelope
principle as \S\ref{sec:lattice}: for the lower endpoint every
contestant sits at its most favourable retained conditional location,
for the upper at its least favourable, and the idiosyncratic scale
alone sets the bisection.\footnote{The obvious shortcut is an
independent-marginal proxy using each runner's total standard deviation
at its unconditional location. This fails exactly where these kernels are
interesting. On a $400$-runner cluster at correlation $0.99$ the
proxy's window dropped $28$ percent of the winner mass, asymmetrically
across two loading groups, and silent normalization returned group
shares of $0.68/0.32$ where the symmetry of the construction forces
$0.50/0.50$.} The raw lattice mass is therefore treated as a diagnostic
rather than a nuisance: the kernels check it and raise an error on a
material defect instead of normalizing it away.

\paragraph{Nested.} Add one global factor with couplings $g_i$ and a dial
$\gamma\in[0,1]$: conditional on the global draw $f_q$ the model is a block race
at shifted abilities, so
$p=\sum_q w_q\,p_{\mathrm{block}}(\mu+\gamma g f_q)$, a finite mixture of block
races. At $\gamma=0$ the clusters are isolated, at $\gamma=1$ fully
coupled, and the covariance contribution scales as $\gamma^{2}gg'$.

\subsection{Trees}\label{sec:tree}

Let a rooted tree have the leaf clusters at its bottom and internal nodes $t$
above, and give node $t$ a scalar effect $b_t\sim N(0,1)$ applied with uniform
strength $\lambda_t$ to every leaf beneath it:
\begin{equation}
X_i=\mu_i+v_i\,a_{c(i)}+\sum_{t\,\in\,\mathrm{anc}(c(i))}\lambda_t b_t
    +\sqrt{D_i}\,\epsilon_i ,
\end{equation}
with the leaf-cluster construction of \S\ref{sec:blocks} restricted to
scalar loadings, $v_i\in\R$ and $a_c\sim N(0,1)$: the shipped tree
kernels take rank-one leaf-cluster effects, where the block kernels
support rank $r$, and they refuse a rank-$r$ leaf loading.

The implied covariance is hierarchical. Two contestants share the variance of
exactly their common ancestors. The root's effect is common to every contestant; hence, it
is choice-irrelevant (this is the recurring invariant), and the pricing kernel gauge-fixes it to zero.\footnote{However, a constructor that reports a covariance may still carry a root variance so
that the raw matrix matches a stated target; the prices are the same
either way.}

The uniform strength is the crucial restriction,
because a shared effect that hits every subtree member equally moves the whole
subtree's field by a translation of its argument. Messages therefore remain
functions of one variable. Per-leaf loadings on internal effects would force
two-dimensional message tables.

The algorithm is two passes over the same lattice.

\paragraph{Upward.} Each leaf cluster carries its block factor
$G_c(x)=\sum_q w_q e^{S_c(x,a_q)}$ from \eqref{eq:blockfield}. Each internal
node, visited deepest first, combines its children under its own effect:
\begin{equation}\label{eq:up}
G_t(x)\;=\;\sum_q w_q \prod_{c\,\in\,\mathrm{children}(t)} G_c\!\left(x+\lambda_t
a_q\right).
\end{equation}
The interpretation: $G_t(x)$ is the probability that every leaf below $t$ survives past
$x$, with $t$'s effect integrated by quadrature and each child's message shifted
because the effect moves that child's whole subtree at once. Shifted evaluations
are linear interpolations on the lattice.

\paragraph{Downward.} The root's cavity is one. Each node, visited shallowest
first, receives the field of everything outside its subtree:
\begin{equation}\label{eq:down}
R_t(x)\;=\;\Bigl[\sum_q w_q\,R_{\mathrm{parent}(t)}\!\left(x-\lambda_{\mathrm{parent}(t)}
a_q\right)\Bigr]\ \prod_{s\,\in\,\mathrm{siblings}(t)} G_s(x).
\end{equation}
The sibling messages sit outside the parent-factor quadrature. This
placement is correct but not obvious: the parent effect shifts the
candidate and every sibling subtree by the same amount, so it cancels
from every candidate-versus-sibling comparison; only the cavity from
outside the parent must be shifted and averaged. Contestant $i$ then races against its own leaf cavity,
$p_i=\int h_i(x)\,R_{c(i)}(x)\,dx$ with $h_i$ as in \eqref{eq:cavity}. The cost
is $O(nLQ)$ plus $O(L\,Q)$ per tree node.\footnote{We checked that a depth-one tree with zero strengths reproduces the rank-one block race to machine
precision. We ran a long $2^{22}$-path Monte Carlo and found that the recursion was
accurate as far as this could determine on every resolvable
entry at depths one through four (root-mean-square $z$ between $0.75$ and
$1.1$ noise level). Entries below the resolution agreed with minimax tilting
\citep{botev2017} to $0.3$ percent relative at probabilities down to
$7\times10^{-8}$, which is within the tilting estimate's own error.}

\begin{remark}[A connection to top-down portfolio construction]\label{rem:hrp}
Long-only portfolio weights can be read as choice probabilities, and we make the further connection to dendrogram-inspired wealth bisection (the hierarchical risk parity construction of \citet{lopezdeprado2016}).
There is a tree race whose implied correlation is exactly the
\emph{floored} cophenetic correlation of a dendrogram: for a linkage
built from the correlation distance $d_{ij}=\sqrt{(1-\rho_{ij})/2}$,
invert the distance transform at each merge height and floor at zero. Give internal node $t$ strength
$\lambda_t^2=\rho_t-\rho_{\mathrm{parent}(t)}$, with
$\rho_t=\max(1-2h_t^2,0)$ at merge height $h_t$, and leaf $i$
idiosyncratic variance $1-\rho_{\mathrm{first}(i)}$.\footnote{The floor at zero is not optional. Shared effects contribute nonnegative
correlation. A tree race cannot represent a negative cophenetic
correlation and dropping the floor pushes the other implied correlations
above one. The shipped \texttt{from\_linkage} constructor keeps the root at its
cophenetic value, which reproduces the floored cophenetic matrix exactly;
this is the covariance-reporting choice of \S\ref{sec:tree}, and the
pricing kernel ignores the common component regardless. Setting the root
effect to zero instead gives the same choices, since a common effect
cannot move an argmin, but different raw correlations.}
\end{remark}

Concentration limits stated on a dendrogram therefore become contest
inversions. We do not pursue the application here.

\section{Jacobians, tie densities, inversion}\label{sec:jac}

We now arrive at the central object of the algorithm: the Jacobian. An outcome of the race is a vector $x$ of finishing values, one per
contestant. Contestant $i$ wins on the set
\[
  R_i=\{x:\;x_i\le x_k\ \text{for every }k\},
\]
a convex cone: the outcomes in which nobody beats $i$. These cones tile
the space, one per contestant, and each win probability is the mass the
joint density puts on its own cone,
\[
  p_i(\mu)=\int_{R_i}q_\mu,\qquad q_\mu(x)=q_0(x-\mu).
\]
Where two cones meet, two contestants dead-heat ahead of the field. So we refer to the
shared boundary
\[
  F_{ij}=\{x:\;x_i=x_j\le x_k,\ k\ne i,j\}
\]
as the photo-finish face between $i$ and $j$.

Now differentiate. Abilities enter the density only as a translation, so
raising $\mu_j$ slides the density rather than reshaping it, and sliding
the density one way is the same as sliding the coordinate the other:
$\partial q_\mu/\partial\mu_j=-\,\partial q_\mu/\partial x_j$. The region
$R_i$ does not depend on $\mu$ at all, so for $j\ne i$ the whole
derivative falls on the integrand,
\[
  \frac{\partial p_i}{\partial\mu_j}
   =-\int_{R_i}\frac{\partial q_\mu}{\partial x_j}\,dx .
\]

We recognize this as a derivative integrated over a region, so the
divergence theorem applies to the vector field $q_\mu e_j$. This lets us
compute over the boundary instead of the volume:
\[
  \frac{\partial p_i}{\partial\mu_j}
   =-\int_{\partial R_i}q_\mu\,n_j\,dS ,
\]
with $n$ the outward unit normal.

However for this we take $q_0$ continuously differentiable with $q_0$ and
$\nabla q_0$ integrable and $|x|^{n-1}q_0(x)\to0$, conditions any
Gaussian law satisfies outright. Integrability alone would
license neither differentiating under the integral sign nor discarding
the flux at infinity. With them the contribution over the sphere at
infinity vanishes, and only the faces of $R_i$ remain.

The observation: those faces are the photo-finish surfaces $F_{ik}$, one
for each rival $k$, and on $F_{ik}$ the outward normal is
$(e_i-e_k)/\sqrt2$. Its $j$-th component is zero unless $k=j$. Every face
therefore drops out except $F_{ij}$. What is left is the statement the
proposition makes precise: slowing $j$ pushes probability out of $j$'s
winning region into $i$'s, and all of it crosses at the surface where the
two dead-heat.

\begin{proposition}[Photo-finish derivatives are a boundary flux]
\label{prop:tie}
For $j\ne i$,
\begin{equation}\label{eq:flux}
\frac{\partial p_i}{\partial\mu_j}
 \;=\;\frac{1}{\sqrt2}\int_{F_{ij}}q_\mu\,dS\;\ge\;0 ,
\end{equation}
the probability flux across the surface on which $i$ and $j$ dead-heat
ahead of the field. Parametrizing that face by the common value
$x_i=x_j=t$ contributes a surface factor $\sqrt2$ which cancels, giving,
with $\varphi_{ij}$ the joint density of $(X_i,X_j)$,
\begin{equation}\label{eq:tiecorr}
\frac{\partial p_i}{\partial\mu_j}
 =\int \varphi_{ij}(t,t)\,
   \Pr\{X_k>t\ \forall k\ne i,j \mid X_i=X_j=t\}\,dt ,
\end{equation}
and under independence
\begin{equation}\label{eq:tie}
\frac{\partial p_i}{\partial \mu_j}
  \;=\; \int f_i(x)\,f_j(x)\prod_{k\ne i,j}S_k(x)\,dx .
\end{equation}
Under conditional independence the same factorization holds node by node
and is averaged over the shared variables. Each row of the Jacobian sums
to zero.
\end{proposition}

Everything else follows from the geometry. Symmetry is $F_{ij}=F_{ji}$:
one surface, one number, so $w_{ij}=\partial p_i/\partial\mu_j$ is
symmetric and nonnegative. The zero row sum is translation invariance, a
conservation law: what leaves one cone enters its neighbours. Writing $B$
for the oriented incidence matrix of the photo-finish graph and $W_e$ for
the diagonal of edge conductances $w_{ij}$,
\[
   J=-B^{\top}W_eB,
   \qquad
   (Jh)_i=\sum_{j\ne i}w_{ij}\,(h_j-h_i),
\]
the discrete analogue of $-\operatorname{div}(a\nabla h)$: $Bh$ takes
differences across photo-finish boundaries, $W_eBh$ is the flux across
them, $B^{\top}W_eBh$ the net flow into each winning region. The
mass-one identity is the same theorem in dimension one, since
$\sum_i f_i\prod_{j\ne i}S_j = \tfrac{d}{dt}(1-\prod_j S_j)$
integrates to one; this is why the implementation's unnormalized mass
checks are diagnostics rather than cosmetics, a material defect meaning
the lattice missed the region, not that normalization was needed.

\subsection{Related literature}

The Jacobian structure has prior art we were unaware of when this work was done.
\citet{muller2022}, differentiating the Williams--Daly--Zachary
gradient identity $\partial E/\partial u_i=\Pr\{i \text{ wins}\}$ for
general joint laws, derive the cross-Hessian as the density that two
alternatives tie for the maximum, and observe that the resulting matrix
is symmetric with signed off-diagonals and zero row sums.

They draw no graph and build no
algorithm from it, their object being convexity moduli for
prox-functions, and their cross-derivative remains an $(n-2)$-fold
integral of the joint density.

Our face parametrization reduces
the same object to a one-dimensional integral of a pair density against
the conditional survival of the field, which is the form in which the lattice works.

The two paths are valid. We prefer to keep the literal interpretation: probability
flux through the boundary between winning regions. This gives the formula, the symmetry, the conservation law, and the graph structure in one geometric picture. On the computational
side, \S\ref{sec:jaccomp}'s identity lets the dense matrix act in
linear work, as we now discuss.

\subsection{Jacobian computation}\label{sec:jaccomp}

Scalability rests on never materializing the matrix. Full Newton needs
the Jacobian only as an operator applied to a vector; a likelihood score
needs one row of it; and the shipped inversion below needs only the
own-coordinate preconditioning slopes.

Per factor node, write
$A_i(x)=f_i(x)/S_i(x)$ for contestant $i$'s hazard and
$G(x)=\prod_k S_k(x)$ for the field. The tie density factors through the
field,
\[
w_{ij}\;=\;\int G(x)\,A_i(x)\,A_j(x)\,dx ,\qquad j\ne i,
\]
so the apparently dense Laplacian is a continuous sum of rank-one
contributions, and its action collapses: substituting into
$(Jh)_i=\sum_{j\ne i}w_{ij}(h_j-h_i)$, the $j=i$ terms cancel between
the two sums, leaving
\[
(Jh)_i\;=\;\int G(x)\,A_i(x)
 \Bigl[\sum_{j} A_j(x)\,h_j\;-\;h_i\sum_{j} A_j(x)\Bigr]dx .
\]
The two inner sums are shared across every $i$, so the whole product
costs $O(nL)$ per node --- the same order as pricing --- rather than the
$O(n^2L)$ of assembling edges. This, more than the Laplacian structure
itself, is what turns the
tie-density Hessian into a scalable object. Observe that a graph with $\binom{n}{2}$
potential edges is applied without ever being constructed.

Both a single entry and a
Jacobian--vector product come out of the same field pass that produced
the probabilities, by the identity above (in the log domain, with
hazards clipped as \S\ref{sec:lattice} describes). Forming all $n^{2}$
entries is possible, but it is a choice with $\Omega(n^{2})$ cost in
output alone, and nothing downstream requires it.

\paragraph{Numerical considerations} In keeping with our
previous lattice remarks we delineate four objects. The first is the continuum
derivative, equation~\eqref{eq:tie}, the exact derivative of the integral
that defines $p$. The second is that formula evaluated by quadrature on a
lattice, with the diagonal imposed from the zero-row-sum identity. The
third, call it the grid derivative, is the exact derivative of the
normalized rectangle sum \emph{holding the selected grid fixed}: with
unnormalized masses $a$, total $T=\mathbf{1}^{\top}a$ and $A=\partial
a/\partial\mu$, it is $(A-p\,\mathbf{1}^{\top}A)/T$. The fourth is the
exact derivative of the full adaptive map, which adds a grid-motion term,
since the window itself depends on $\mu$.

The factor implementation returns the third. It differentiates the rectangle sum
itself, then applies the normalization quotient: a directional derivative
$v$ of the unnormalized masses is returned as
$(v-p\,\mathbf{1}^{\top}v)/T$. It does not differentiate the adaptive
window, and we measured the residual arising from that omission.

That quotient is critical. Without it, the returned directional
derivative disagrees with central finite differences of the returned map
by $2\times10^{-3}$ at $L=25$ lattice points. With it the two agree to
$3\times10^{-11}$ from $L=101$ upward, which is the production range.

At $L=25$ a residual of $6\times10^{-4}$ survives, and it is not a defect
in the quotient. It is the grid-motion term, the difference between the
third object and the fourth: a finite difference recomputes the window at
each perturbed $\mu$ and therefore sees a contribution the fixed-grid
derivative does not contain. It falls below the agreement figure by
$L=101$, so in the production range the grid derivative serves as the
derivative of the adaptive map to the accuracy reported.

The materialized Jacobians of the block, nested and tree kernels are
the \emph{second} object, not the third. They are continuum tie densities
evaluated by quadrature with the diagonal imposed, ignoring both the
normalization quotient and the window. They are consistent estimates of
the continuum derivative, and they are not exact derivatives of the
normalized discrete maps.\footnote{A caller who needs the
latter (an optimizer, say) should difference the forward map.}

The continuum form is implemented alongside the grid form because it is
symmetric by construction and can be used as a preconditioner without
further checking. The grid form is symmetric only to the accuracy of its
row sums, for the same reason: at $25$ points it is off by
$3\times10^{-4}$ and at $65$ by $7\times10^{-11}$, but at the $257$-point
default the two forms agree to machine precision, so the distinction is
one of guarantee rather than of observed values in production.

A comment on finiteness. The photo-finish integrand is
factored as $$[f_i e^{L-\log S_i}]\times[f_j/S_j].$$ The first factor is
bounded by a density. The second is a hazard rate, which grows linearly
under the normal base and exponentially under the Gumbel base; the
exponent clip absorbs either. The symmetric square-root split, which
looks more natural, overflows wherever survivals vanish.

As a final remark, the block Jacobian is exact for
rank-one cluster loadings. The nested Jacobian is the
mixture of block Jacobians. The tree Jacobian is exact within a cluster
and Gram-approximate across clusters. That approximation is safe for
Newton because the residual is always computed from the exact forward
map, so an approximate Jacobian can slow convergence but cannot change the
answer to which it converges.

\subsection{Inversion}

The inversion rests on a convex program.

\begin{theorem}[The inverse exists and is unique on contrasts]\label{thm:inverse}
Fix a covariance of full support (any member of the grammar with every
$D_i>0$) and let $W(\mu)=\E\min_i(\mu_i+\epsilon_i)$,
$\epsilon\sim N(0,\Sigma)$. Then $W$ is concave and differentiable with
$\nabla W(\mu)=p(\mu)$, and its Hessian is $-L(\mu)$, the negative of the
graph Laplacian whose edge weights are the correlated tie densities
\eqref{eq:tiecorr}, reducing to \eqref{eq:tie} under independence.
Every tie density is strictly positive, so $W$ is strictly concave on the
contrast space $\mathbf{1}^{\perp}$. For any target with $p^{\star}_i>0$
and $\sum_i p^{\star}_i=1$, the program
\[
\max_{\mu}\; W(\mu)-\langle p^{\star},\mu\rangle
\]
is constant along $\mathbf{1}$, strictly concave and coercive on
$\mathbf{1}^{\perp}$, and its unique mean-zero maximiser is the inverse:
$p(\mu^{\star})=p^{\star}$. If some $p^{\star}_i=0$ the supremum is
approached only as the contrast
$\mu_i-\min_{j:\,p^{\star}_j>0}\mu_j\to+\infty$ (the raw coordinate
statement depends on the gauge); no finite inverse exists, and any
floored target yields a finite point on that diverging contrast.
\end{theorem}

\begin{proof}
A minimum of affine functions of $\mu$ is concave and expectation
preserves this. Ties have measure zero under full support, so the
minimiser is almost surely unique, and Danskin's theorem
\citep{danskin1966} gives
$\partial W/\partial\mu_i=\Pr\{i\ \text{wins}\}$ --- the identity known
in discrete choice as the Williams--Daly--Zachary theorem
\citep{mcfadden1981}.\footnote{Danskin's
theorem differentiates a minimum over a parameter. For
$W(\mu)=\E\min_i(\mu_i+\epsilon_i)$ the naive move, differentiating
inside the expectation, is illegitimate because the minimum is not
differentiable where two contestants tie. The theorem says the
derivative exists anyway whenever the minimiser is unique, and equals the
derivative of the selected branch: here $\partial/\partial\mu_i$ of
$\mu_i+\epsilon_i$, which is one, on the event that $i$ is the
minimiser and zero otherwise. Taking expectations gives the win
probability. Full support is what supplies uniqueness, since ties then
have measure zero; without it $W$ has kinks and only a subdifferential.
This is the step that makes $p$ a gradient, and hence the inversion a
convex program rather than a root-find.}

The Hessian is given by Proposition~\ref{prop:tie}, which assumed no
independence: its off-diagonals are the boundary fluxes
\eqref{eq:tiecorr} under the full correlated law, and its rows sum to
zero by the same translation invariance $W(\mu+c\mathbf{1})=W(\mu)+c$
that fixes the diagonal without further computation. The Hessian is
therefore $-L(\mu)$ with $L$ the weighted graph Laplacian of the
photo-finish graph; under full support ($D_i>0$, so every pairwise difference has
positive variance) each edge weight is strictly positive, the tie graph
is complete, and $L\succ0$ on $\mathbf{1}^{\perp}$: $W$ is strictly
concave there.

Coercivity follows from the elementary bound
$W(\mu)\le\min_i\mu_i$: for each $i$,
$\min_k(\mu_k+\epsilon_k)\le\mu_i+\epsilon_i$, and $\E\epsilon_i=0$.
For a nonconstant unit contrast $d$ and interior target,
$\langle p^{\star},d\rangle>\min_i d_i$, because $p^{\star}$ puts
positive mass on a coordinate where $d$ exceeds its minimum; hence
\[
W(td)-t\langle p^{\star},d\rangle
 \;\le\; t\bigl(\min_i d_i-\langle p^{\star},d\rangle\bigr)
 \;\longrightarrow\;-\infty
\]
linearly in $t$. The decay is uniform, not merely raywise:
$d\mapsto\langle p^{\star},d\rangle-\min_i d_i$ is continuous and
strictly positive on the compact unit sphere of $\mathbf{1}^{\perp}$,
so it attains a positive minimum $\kappa$, and the displayed inequality
at $t=\lVert\mu\rVert$ gives $F(\mu)\le-\kappa\lVert\mu\rVert$ for every
$\mu\in\mathbf{1}^{\perp}$, writing $F$ for the objective, which is
coercivity for arbitrary diverging sequences. (The bound is on $F$
itself, not on $F$ relative to $F(0)$: the objective may climb toward
its maximizer before the linear decay takes over.) Constancy along
$\mathbf{1}$ is translation invariance
again, exactly when the target sums to one. A strictly concave coercive
function on $\mathbf{1}^{\perp}$ attains a unique maximum, where the
gradient vanishes: $p(\mu^{\star})=p^{\star}$.

For the boundary claim, suppose $p^{\star}_i=0$ for $i$ outside
$S=\{j:p^{\star}_j>0\}$. Under full support $p_i(\mu)>0$ at every finite
$\mu$, so no finite point inverts the target. Invert the target
restricted to $S$, then send each excluded $\mu_i\to+\infty$ holding the
$S$-contrasts fixed: by dominated convergence the race converges to the
restricted race and the objective to its restricted optimum, so the
supremum is approached along exactly those diverging contrasts.
Conversely, on any maximizing sequence an excluded contrast that
remained bounded would admit, after gauge fixing, a convergent
subsequence whose limit gives alternative $i$ strictly positive
probability, and the objective, continuous in $\mu$, would sit strictly
below the restricted optimum there --- contradicting maximality.
\end{proof}

\subsection{Related work}

We remark on the connection to economics. We are considering the Legendre dual of the social surplus function of
\citet{mcfadden1981}. Share inversion has frequently been viewed as a convex problem: \citet{berry1994} inverts market shares in
logit-family demand, \citet{galichon2022} identify the general random-utility
inversion with optimal transport, \citet{chiong2016} compute it by
linear programming in the discrete case, and \citet{muller2022} state
the inversion as the subdifferential of the conjugate surplus,
$u\in\partial E^{*}(p)$, en route to their prox-functions.

Closest is
\citet{li2018}, who for arbitrary non-atomic bases writes the inversion
as an unconstrained convex minimization whose objective gradient is the
share map, with globally convergent superlinear methods. That is a formulation
of the optimizer, but without a fast evaluator of the shares
and their Jacobian at large $n$, which is what we supply. The shared field is that evaluator, so
the two are complements. What the field representation adds is not the
convex program but its oracles: the shipped solver preconditions with
the field's own-coordinate slopes. The Hessian is available as a matrix-free operator at $O(nLQ)$ per application for Newton--Krylov methods.

\subsection{Fitting objective}

Because choice probabilities are
invariant to a common utility shock (as noted the model depends on $\Sigma$ only
through $P\Sigma P$ with $P=I-\mathbf{1}\mathbf{1}^{\top}/n$) a factor
proportional to $\mathbf{1}$, or a tree-root effect applied to every leaf, will not change any choice probabilities. Therefore the right fitting objective for a dense covariance
is the {\em projected} residual,
\[
\min_{V,\,D\ge0}\;\bigl\lVert P\,(\Sigma-VV'-\mathrm{diag}\,D)\,P\bigr\rVert_F^2,
\]
not a fit of $VV'+\mathrm{diag}(D)$ \emph{to} $P\Sigma P$ (the projection
of a rank-plus-diagonal matrix is no longer rank-plus-diagonal, so the
second problem has some manufactured error compared to the first).

The reference implementation minimizes that objective in
\texttt{factor\_model\_projected}, alternating two steps. The $V$-step
takes the top $k$ eigendirections of the current residual in the contrast
basis. The $D$-step solves for the diagonal exactly, subject to a lower
bound, from the normal equations
\[
  (P\!\circ\!P)\,d=\mathrm{diag}\bigl(P(\Sigma-VV')P\bigr),
\]
whose Gram matrix is closed form, so the step is $O(n)$ rather than a
least-squares solve (\S\ref{sec:general}).

Everything downstream works from projected residuals, so a matrix
already of the form $VV'+\mathrm{diag}(D)$ passes through whenever its
factor rank lies within the fitted grammar's budget --- up to the
common-shock equivalence, factor-quadrature error, and a diagonal floor
of $10^{-6}$ of each runner's own variance.\footnote{Where
the introduction speaks of an identified class, it means the restriction
that is choice-relevant and computable on $\mathbf{1}^{\perp}$: choices
depend on $\Sigma$ only through $P\Sigma P$, an observational-equivalence
statement, not a claim about what an estimation design can recover.
Covariances outside the grammar are perfectly good models; this grammar
simply does not price them.}

\subsection{Inversion algorithm}

The implemented inversion (\texttt{abilities\_from\_race}) runs as follows.

\begin{center}
\fbox{\parbox{0.92\linewidth}{%
\textbf{Inversion} (all grammars; min-wins). Given a target $p^{\star}$
with positive entries summing to one, set $t=\log p^{\star}$ and
initialize $\mu^{0}=-(t-\bar t)/2$. Repeat until
$\max_i|r_i|<10^{-8}$:
\begin{enumerate}\setlength{\itemsep}{0pt}
\item one field pass: $\hat p=p(\mu)$, and negative preconditioning
slopes $\tilde s_i$ (factor grammar: the own-coordinate grid slopes
$\partial p_i/\partial\mu_i$ of the same pass; block/nested/tree: the
own-slopes of an independent race at matched total variances, one extra
$O(nL)$ pass --- an approximation used only to precondition, never to
form the residual);
\item log residual $r_i=\log\hat p_i-t_i$; log-domain slope
$d_i=\min(\tilde s_i/\hat p_i,\,-10^{-6})$;
\item damped, clipped preconditioned step
$\mu_i\leftarrow\mu_i-\operatorname{clip}(\alpha\,r_i/d_i,\pm2)$,
then recenter $\mu\leftarrow\mu-\bar\mu$.
\end{enumerate}
Damping: $\alpha=1$ ($0.7$ at $n=2$, where the undamped update on the
$K_2$ photo-finish graph is a two-cycle); for the structure grammars
$\alpha=0.7$, halved whenever the residual fails to shrink.}}
\end{center}

Each sweep is a diagonally preconditioned residual iteration on the
log-probability residual
$r_i=\log p_i(\mu)-\log p^{\star}_i$.\footnote{In handicapper's terms,
for a million runners: you know the win probability you want for each
runner and must adjust handicaps until the model produces it. Each pass
of the reported solver computes every runner's current win probability
and every runner's \emph{own} slope --- if only this runner slows
slightly, how much does its own probability move? --- then adjusts each
runner by its personal sensitivity, re-prices the whole race, and
repeats.}

Interactions are not ignored (every re-pricing includes them);
they are just not consulted when choosing each step. The $Jv$ oracle is
the more powerful tool: propose changing \emph{all} abilities at once by
$v$ and it returns the resulting change in \emph{all} win probabilities,
without constructing the million-by-million Jacobian's $10^{12}$
entries. A Newton--Krylov solver would repeatedly ask it what
coordinated change removes the current errors most efficiently. That
more sophisticated solver is not what produced the reported
million-runner result; the own-slope iteration did, and making
Newton--Krylov beat it is future solver work.

Thus in contrast to a full Newton step solving a linear system in the Jacobian, here each
coordinate is instead divided by a preconditioning slope $\tilde s_i$,
which for the factor grammar is the own-coordinate slope, and for the other grammars an independent-race
approximation to it. The residual, by contrast, always comes from the
exact structured forward map, so the preconditioner does not change the
root equation; it can affect convergence and the finite-tolerance point
returned.

Working in logs rather than in probabilities matters because win probabilities span
many orders of magnitude, and an absolute residual would effectively
ignore everything but the front of the field. A step clip and a
recentering keep the iterate in the contrast space.

Theorem~\ref{thm:inverse} says the inverse exists and is unique for the
Gaussian grammars, which is where it is stated. For a general base the
same proof goes through when the location family has finite first
moments, an everywhere-positive sufficiently regular density and no
atoms: concavity and coercivity of $W$ use nothing Gaussian, and strict
concavity needs only positive tie densities, which full support supplies.
We have not written that generalization out formally, so for non-normal
bases the uniqueness claim should be read as proved under those
conditions rather than universally. The theorem also does not say that
this clipped iteration finds the inverse.

\subsection{Rationale for defaulting to a diagonal solver}

The solver's convergence, as implemented, stands on empirical evidence, which is strong on the
benchmark families but not universal. We can try to break the algorithm and construct strongly
correlated fields on which the damped iteration stalls at a residual
near $10^{-4}$ instead of converging.

A matrix-free Newton--Krylov solver using the exact grid Jacobian--vector product is
provided also but is work in progress. Naively assembled it lost badly to the damped diagonal iteration\footnote{
$387$ seconds without converging against $3.7$ seconds converged at
$n=200$.}. Assembled properly (log-residual system symmetrized by the
$\sqrt{p}$ scaling and using integration-by-parts form of the product,
conjugate gradients preconditioned by the own-slopes and also a null vector
projected out) it
does converge on the benchmark fields ($9.6\times10^{-10}$, but at roughly
sixty times the wall clock). On the constructed stall case it
diverges where the diagonal iteration stalls.

So, for this reason the diagonal iteration
is the production solver we recommend, whereas a robust Newton--Krylov remains solver
engineering. The nearest
proven result we know for a problem of this shape is the damped Newton
method of \citet{kitagawa2019} for semi-discrete optimal transport:
prescribe the mass each Laguerre cell must capture and adjust one
weight per cell, with a Hessian assembled from source mass on shared
cell boundaries much as our Jacobian is assembled from tie density on
shared win-region boundaries. They prove global linear convergence
under regularity conditions on the cost and source.
\citet{levy2021} run the method at ten million cells. Whether a
race-market analogue of their conditions yields a convergence proof
here is open.

We tested the round trip at $n=400$, with twenty
clusters and a tree of depth three (the committed \texttt{bench.py
invert}). Timing for independent was $0.04$\,s, tree $0.37$\,s, and blocks $0.67$\,s. A factor
rank-two took $1.7$\,s and nested $3.7$\,s. In each case the maximum log-probability residual was below $10^{-8}$.

The implementation does not hide errors. A zero or negative entry raises
rather than being quietly repaired, since by Theorem~\ref{thm:inverse} it
has no finite inverse. Flooring small entries is available but must be
requested by the user. For an entry whose target is exactly zero, the
theorem says every finite-floor solution is a finite point on a contrast
that diverges as the floor is removed, so the floored answer understates
that contrast in the direction the theorem indicates. For entries that
are merely small and positive, reading the floored-and-renormalized
solution as a coordinatewise bound on the true contrasts would need a
monotonicity argument we do not make here.\footnote{At time of writing
the diagnostics name which entries were floored, whether the iteration
converged, and the residual it reached; non-convergence warns.}

Every grammar gets that same contract. The target is validated once,
before the call splits by structure, and all five paths return through a
single exit, so a block or tree inversion reports floored entries,
iteration count and convergence exactly as the factor path does.

Structure matters at inversion time, and not by a little. Inverting
block-generated probabilities under an assumed independent model
mislocates abilities by up to three tenths of the field's spread.\footnote{At
$n=200$ with thirty clusters and within-cluster correlation $0.65$, the
largest mislocation is $2.49$ against a spread of $8.24$, with a median
of $0.34$.}

\section{Approximate results for dense covariance}\label{sec:general}

We ship a fitting procedure as follows for the general case. It comes with several caveats. We use the following notation. Let $$s_i=\sqrt{\Sigma_{ii}}$$,
$S=\mathrm{diag}(s)$, and $$C=S^{-1}\Sigma S^{-1}$$ the correlation
matrix. Every fit, residual and closing solve below is carried out in
$\Sigma$. Correlations appear only once, as the distance used for
clustering in stage~3.

That choice is apparently forced, as we discovered. The tempting
alternative is to standardize to $C$,
fit there, and restore the scales afterwards, which would let the fitter
work on a matrix with unit diagonal. It did not work, for a reason related to the same under-specification of covariance we have encountered several times already.

Choices depend on $\Sigma$ only through $P\Sigma P$, so
the residual that matters after restoration is $PS(C-M)SP$, not
$P(C-M)P$, and those differ because $P$ does not commute with $S$ unless
every scale is equal. Put another way: the direction choices ignore is
the constant vector in covariance coordinates, but $S^{-1}\mathbf{1}$ in
correlation coordinates. Fitting in correlation space with the ordinary
projection removes the wrong direction.\footnote{This really matters. Take $n=2$, $C$ equicorrelated with $\rho=0.9$, and
scales $s=(1,2)$. With two alternatives the only choice-relevant
quantity is the variance of the difference $X_1-X_2$, since that alone
decides the race. In correlation coordinates the fit $M=(1-\rho)I$ scores as exact.
Restored to covariance coordinates it gives a difference variance of
$0.5$ where the truth is $1.4$, wrong by nearly a factor of three, and
the correlation-space objective reports no error at all.} A single global rescaling of $\Sigma$ for conditioning is harmless but
unequal scales break the commutation.\footnote{With this issue resolved, the test suite pins a two-alternative example, and a round trip on a covariance built as
$VV'+\mathrm{diag}(D)$ exactly, with idiosyncratic variances drawn across
a hundredfold range. Pricing that matrix through the dense front door
must reproduce the race obtained from the $(V,D)$ it was built from.
Nothing there is approximate, so any disagreement is the coordinate
choice rather than the grammar.}

\subsection{Heuristic fitting algorithm}

The stages, in order:
\begin{enumerate}
\item Global factors: the quotient fit of \S\ref{sec:jac}
(\texttt{factor\_model\_projected}), $k=3$ throughout, returning $(V,D_0)$
by alternation. Each half-step is optimal in its own block, so the projected objective descends monotonically to a
coordinatewise-stationary point.\footnote{Above $n=800$ the $V$-step is subspace
iteration rather than a full eigendecomposition and its subspace is only
approximate.} The objective is
evaluated each sweep at $O(n^2)$. A sweep that fails to improve it is
discarded, and the alternation returns the best iterate.
\item The working residual is then
projected once, $R=P\,(\Sigma-VV'-\mathrm{diag}\,D_0)\,P$: the raw residual
still contains the choice-irrelevant common component the quotient fit
rightly ignored, and letting later stages chase it manufactures projected
error (measurably: the raw-residual variant distorts an input already of the
form $VV'+\mathrm{diag}(D)$ by $1.7\times10^{-2}$ of total variation;
the projected pipeline returns it at the $2\times10^{-4}$ node-noise
floor).
\item Blocks: average-linkage clustering on
$\sqrt{(1-C)/2}$ cut to at most $20$ clusters; for each cluster, the leading
eigenvector of the within-cluster block of $R$ with its diagonal zeroed,
one rank-one loading per cluster (clusters are disjoint).
\item Residual promotion: the top $m=5$ eigencomponents of what remains of $R$, diagonal
zeroed, appended as further factor columns; negative eigenvalues are not
promoted; numerically dead columns are dropped.
\item The closing diagonal
solves the projected problem's own normal equations,
$(P\!\circ\!P)\,d=\mathrm{diag}(P(\Sigma-V_{\mathrm{all}}V_{\mathrm{all}}')P)$
subject to $d_i\ge\ell_i$, where the floor is relative to each runner's
own variance, $\ell_i=10^{-6}\,\Sigma_{ii}$. Relative matters: an
absolute floor at a multiple of the \emph{mean} variance destroys nearly
singular contrasts, and $\Sigma=\mathrm{diag}(10^{-8},10^{-8},1)$ makes
the failure vivid --- it forced $\mathrm{Var}(X_1-X_2)$ from
$2\times10^{-8}$ up to $6.7\times10^{-4}$, turning a near-certain
head-to-head into a coin flip (choice error $0.48$) while the projected
covariance residual sat one hundred times below the warning threshold. The Gram is exactly
$(1-\tfrac2n)I+\tfrac1{n^2}\mathbf{1}\mathbf{1}^{\top}$, so the
constrained solve is closed form in $O(n)$: substituting $d=\ell+x$
componentwise with $x\ge0$, writing $\ell$ for the vector of
per-runner floors, leaves the same Gram with a shifted right-hand side,
and the KKT conditions make the active set a threshold set, solved by
water-filling. It is the constrained minimizer, not an
unconstrained solve with a floor applied afterwards, which the Gram's
off-diagonal coupling would leave suboptimal. At $n=2$ the Gram is rank
one, only $d_1+d_2$ is identified, and the symmetric representative is
returned.
\item A second candidate, a pure eigenfit at the same total rank, closed by the same
diagonal solve, is computed, and whichever candidate leaves the
smaller choice-relevant residual is kept, the pipeline winning ties:
the greedy factor-plus-blocks allocation is the wrong shape for
globally smooth covariance, and the kernel battery below measures the repair.
\end{enumerate}

The constants $k$, the cluster count and $m$ are
fixed, not tuned per ensemble; a spectral rule for $m$ was tested and
rejected (it saturates its cap for marginal gains and does not repair the
failure case below). What results is a factor-plus-blocks race priced in
one seeded, reproducible call (deterministic under Gauss--Hermite
nodes, seeded scrambled Sobol past the tensor budget), and the pipeline
is the shipped one-call
(\texttt{fit\_covariance}; \texttt{race\_probabilities} accepts
\texttt{cov=} directly), not a research script.

The alternation
is nonconvex and its default start can stall on easy inputs (an exactly
in-grammar matrix at $n=12$ with variances spread across two decades sat
at objective $0.88$ where zero was attainable), so when the first
start's objective is materially above zero, two further deterministic
starts are tried --- diagonal-heavy and eigen-residual --- and the best
kept; each rescues that case to $10^{-13}$.

The shipped function also
reports its own quality and warns on failure modes. In particular it
flags a well-fitted matrix whose pricing needs a high-rank, sharp factor
integral. The quadrature, not the fit, is then the binding constraint. A nearly
singular
contrast can also cause problems. The report therefore carries
\[
\max_{i<j}\;
\frac{\bigl|(e_i-e_j)^{\top}R\,(e_i-e_j)\bigr|}
     {(e_i-e_j)^{\top}\Sigma\,(e_i-e_j)} ,
\]
the worst pairwise difference-variance error in units of that
difference variance, which is $0.49$ on the diagonal example above
(where the global residual is $5\times10^{-4}$) and $10^{-8}$ once the
relative floor restores the fit.

The objective is nonconvex in $(V,D)$, so
no global optimization can be assured. The residual report exposes the
value reached.

Multistart dispersion is the natural diagnostic, but it
has to be measured in the right coordinates, and the obvious choice is
the wrong one. Across four ensembles at $n=300$, with eight random diagonal starts
spanning two decades of assumed idiosyncratic share, every start reaches
the same objective to $4\times10^{-16}$ relative, and no random start
ever beats the shipped one. On three of the four the fits agree in
choices as well, to a total variation of $10^{-14}$.

\subsection{Symmetry, and its challenges}

Block equicorrelation is the exception, and it was an instructive one. There the
starts agree on the objective and on $D$ to every digit, and still price
races that differ by a total variation of $0.25$. The cause is degeneracy
rather than optimization. Centering a six-block matrix leaves a five-fold
tied leading eigenvalue, so a rank-three fit selects an arbitrary
three-dimensional subspace of a five-dimensional one. Every selection is
exactly optimal and each implies a different race. Raising the rank to
the multiplicity makes the fit exact and collapses the disagreement to
$8\times10^{-4}$.

These experiments convince us of two things. First, disperse the priced race across starts, not the objective, which is blind to precisely the case that matters. And when
dispersion does appear, the remedy is rank rather than a better
optimizer. So the shipped fitter warns when the chosen rank splits a tied
eigenvalue of the centered covariance, and tells the user which rank that would take
the whole tied group.

The extreme case is also the oldest model in the subject. Under equicorrelation $P\Sigma P$ is a multiple of $P$, so its
spectrum is $(n-1)$-fold degenerate. A great many representations are
exactly optimal. They are all
choice-equivalent, so none can be said to be wrong. But they do not cost the same to
price. A representation carrying three factor columns is priced by a
three-dimensional quadrature, where the equivalent pure diagonal needs
none.

Group structures degenerate the same way. One can arrange equal-sized groups with equal within-group correlation. This leaves one tied leading direction per contrast among the groups, so $g$ equal
groups give multiplicity $g-1$, with equicorrelation the case $g=n$. The
default $k=3$ is therefore exactly right for four equal groups but otherwise wrong, or should we say, inefficient. Measured at $n=24$ and within-group correlation
$0.7$, against the same race priced through the block grammar directly:

\begin{center}\small
\begin{tabular}{rrrr}
\toprule
equal groups & leading multiplicity & rank cuts the tie & TV against grammar\\
\midrule
$2$ & $1$ & no & $2.4\times10^{-4}$\\
$3$ & $2$ & no & $2.8\times10^{-4}$\\
$4$ & $3$ & no & $1.8\times10^{-4}$\\
$6$ & $5$ & yes & $1.0\times10^{-2}$\\
$8$ & $7$ & yes & $3.7\times10^{-2}$\\
$12$ & $11$ & yes & $4.0\times10^{-2}$\\
\bottomrule
\end{tabular}
\end{center}

The first three rows sit at the fitter's ordinary approximation error.
The last three are the rank cutting a tied group, and the cost grows with
how much of the group is cut. The warning fires on all three, and on the
three-group row as well, where the error is harmless: erring toward the
false positive is the right direction here.

To emphasize, all of these issues apply to the dense covariance usage where the package is being asked to fit. A
user can supply a structure they prefer without ambiguity. And symmetry is
presumably known in advance. It is exchangeability among runners that produces the tie, and a
caller who knows the group structure knows the rank to ask for, namely
one less than the number of exchangeable groups.

\subsection{Accuracy}
The one-call fit-and-price takes four
seconds ($0.6$ of them the fit) against thirty-six for a $10^6$-draw
simulation, and agrees with it at its noise floor on the resolved bulk.

The tail is harder to check, and we are careful about what the check
establishes. For the $82$ percent of alternatives with zero simulation
wins, agreement is unverifiable entry by entry. Worse, that set is
selected by the same simulation that would have to certify it, so a
fixed-set binomial bound does not apply to it.

Sample splitting does apply, and we run it as a separate experiment
(its field is configured differently, which is why its zero-win fraction
differs from the one above). Fix the zero-win set on one
half of the draws, then count the wins an independent second half places
anywhere in that set. Because the set is fixed before the second half is
looked at, its count is binomial in the set's true mass, and a
Clopper--Pearson interval is valid.

On a dense $n=2000$ correlation with three global factors and a
forty-block layer, priced in a wide field where $63.6$ percent of
alternatives take no win in the first million paths, the model puts
$1.13\times10^{-4}$ of mass on that set. The second million paths place
$111$ wins in it, a $95$ percent interval of
$[9.13\times10^{-5},1.34\times10^{-4}]$. The model's mass is inside it.

Scored on the second half alone, so that no statistic is evaluated on the
draws that selected its own target, full-vector total variation is
$4.4\times10^{-3}$ against a split-half noise floor of $5.4\times10^{-3}$:
at the referee's own resolution. Tail exactness at small $n$ still comes
from the Botev referee, which prices individual orthant probabilities
rather than counting wins.

\subsection{Testing the choice of objective}

Fitting
$VV'+\mathrm{diag}(D)$ \emph{to} $P\Sigma P$ is the wrong
objective (\S\ref{sec:jac}). With the identified objective and projected residuals throughout, the ensemble
study below ran a raw top-$k$ eigenfit pipeline and the identified
pipeline over twenty seeds per ensemble. On thirteen of fifteen named
ensembles the arms are statistically indistinguishable and none favors
the raw arm. On block equicorrelation the identified arm is eight
times better at the median ($2.5\times10^{-3}$ against
$2.0\times10^{-2}$), with a worst seed of $3.2\times10^{-3}$ against
the raw arm's $0.162$: the raw eigenfit spends its rank on a
near-common component, as the identification argument predicts.

\subsection{Testing against different covariance generation families}

Accuracy depends on the family. We used the named collections of the \texttt{randomcov} library (pinned at commit
\texttt{0d27a51}), over \emph{twenty seeds} per ensemble at $n=300$,
each seed refereed by its own $10^6$-path simulation (the committed
\texttt{run\_ensembles4.py} and \texttt{run\_kernel4.py} regenerate
everything).

The TV columns are full-vector total variation against the empirical
frequency vector, zeros included, so the referee's own counting noise and
any unresolved tail mass contribute to the reported number. That makes
the statistic upward biased in expectation, by convexity, but not a
realization-wise upper bound: on a single seed sampling error can cancel
part of the model error. Only the per-entry column is restricted, to
entries with at least $25$ referee wins. TV is reported as median, ninth
decile and worst over seeds; the last column is the median per-entry
absolute error on those resolved entries.

\begin{center}\small
\begin{tabular}{lrrrr}
\toprule
ensemble & med TV & q90 TV & worst TV & med $|p-\hat p|$\\
\midrule
block equicorrelation & $2.5\times10^{-3}$ & $3.0\times10^{-3}$ & $3.2\times10^{-3}$ & $1.4\times10^{-5}$\\
factor $+$ sparse links & $2.9\times10^{-3}$ & $6.7\times10^{-3}$ & $1.6\times10^{-2}$ & $1.9\times10^{-5}$\\
sparse-precision graphical & $4.5\times10^{-3}$ & $5.6\times10^{-3}$ & $6.0\times10^{-3}$ & $1.5\times10^{-5}$\\
Wishart & $1.2\times10^{-2}$ & $1.4\times10^{-2}$ & $1.5\times10^{-2}$ & $2.6\times10^{-5}$\\
AR(1) & $4.6\times10^{-2}$ & $8.4\times10^{-2}$ & $1.5\times10^{-1}$ & $7.0\times10^{-5}$\\
residuals (dense-strong) & $6.8\times10^{-2}$ & $8.6\times10^{-2}$ & $9.2\times10^{-2}$ & $2.5\times10^{-4}$\\
\bottomrule
\end{tabular}
\end{center}

Re-running any fixed covariance seed moves its score only at the
referee's noise scale; the medians themselves are aggregates over the
twenty distinct covariance seeds, whose genuine spread the q90 and worst
columns report. Eight further named
ensembles (LKJ, onion, vine, elliptope walk, animals, Marchenko--Pastur
and spiked spectra, hierarchical) fall between the third and fifth
rows; Archakov--Hansen sits with AR(1) at $4.9\times10^{-2}$.

\subsection{Kernel covariance as a weakness}

Testing identified a type of covariance where the user should not expect our fitting procedure combined with the lattice technique to yield good results, and they would be advised to use the alternatives.

A stratified battery
(kernel type $\times$ length scale $\times$ budget, twenty seeds per
cell, unit-square inputs) separates two distinct mechanisms that a
single ``kernel'' row had conflated:

\begin{center}\small
\begin{tabular}{llrrr}
\toprule
kernel & length scale & med TV, $m{=}5$ & $m{=}12$ & $m{=}5$, $2^{14}$ nodes\\
\midrule
RBF & $0.08$ & $3.4\times10^{-2}$ & $3.1\times10^{-2}$ & $3.3\times10^{-2}$\\
RBF & $0.2$  & $2.6\times10^{-2}$ & $2.6\times10^{-2}$ & $1.2\times10^{-2}$\\
RBF & $0.4$  & $1.6\times10^{-2}$ & $1.7\times10^{-2}$ & $8.2\times10^{-3}$\\
Mat\'ern-$\tfrac32$ & $0.08$ & $2.9\times10^{-2}$ & $2.5\times10^{-2}$ & $2.4\times10^{-2}$\\
Mat\'ern-$\tfrac32$ & $0.2$  & $2.5\times10^{-2}$ & $2.3\times10^{-2}$ & $1.6\times10^{-2}$\\
Mat\'ern-$\tfrac32$ & $0.4$  & $2.1\times10^{-2}$ & $1.7\times10^{-2}$ & $1.2\times10^{-2}$\\
\bottomrule
\end{tabular}
\end{center}

At short length scale the bottleneck is \emph{representation}.
Correlation is locality, the rank needed grows with the number of
correlation lengths in the domain, and accordingly extra rank helps
($3.4\to3.1\times10^{-2}$) while extra quadrature does not.

At long length scale the fit is not the problem at all. A rank-$27$
eigenfit holds the RBF draw at length scale $0.4$ to a projected residual
of $10^{-3}$. But pricing a well-fitted smooth kernel means a
twenty-seven-dimensional \emph{sharp} factor integral, since the closing
diagonal is nearly zero and the conditional races are nearly
deterministic. There the node budget is the binding constraint:
raising the residual-factor budget from $m=5$ to $m=12$ does little,
while $2^{14}$ nodes halve the error
and $2^{16}$ nodes reach $6\times10^{-3}$ on the diagnostic draw.

Stage~6 of the pipeline is what reaches this regime at all. The greedy
factor-plus-blocks allocation leaves a projected residual of $0.14$ on
the same draw that the matched-rank eigenfit holds to $10^{-3}$, a
misallocation worth up to $7\times10^{-2}$ of total variation, and not a
limit of the grammar.

The dispatch conclusion stands. Genuinely local covariance wants
sequential conditioning with cheap sparse-precision draws, and Mat\'ern
is exactly the family with the SPDE route \citep{lindgren2011}.
Near-singular smooth covariance wants either a larger node budget or
plain simulation. A production front door should dispatch on the
structure it detects, and the shipped call's warnings distinguish the two
regimes.

\section{Benchmarks}\label{sec:num}

\paragraph{The structure-aware per-alternative comparator.} The fairest
rival is not GHK but this paper's own conditioning run the literature's
way: for each alternative separately, integrate the conditional
independent-race integrand over the same Gauss--Hermite factor nodes
(the construction of \citet{butler1982}; \texttt{mvtnorm}'s
\texttt{lpRR}/\texttt{slpRR} are its Monte Carlo form). With the same
nodes and the same lattice the two computations are algebraically
identical, and the measurement
confirms it: at $n=200$, rank~2, the per-alternative
version agrees with the shared field to total variation
$3\times10^{-15}$ and costs $601$ times as much ($22$\,s against
$36$\,ms; the committed \texttt{bench.py bm}). The comparison isolates
the paper's contribution exactly: not the conditioning, which is forty
years old, but the field shared across alternatives.

\paragraph{Against the field.} All-$n$ choice probabilities under the rank-2
factor covariance family of \S\ref{sec:factor}, against a
$2^{20}$-point QMC reference. Tail is the maximum absolute log-error over
reference probabilities in $[10^{-4},10^{-3}]$; the $n=10$ draw has none.
Frequency simulation is run at wall clock matched to the lattice and at ten
times that budget; ``zeros'' counts tail alternatives that received no draws.

\begin{center}\small
\begin{tabular}{lrrrrr}
\toprule
& \multicolumn{2}{c}{$n=10$} & \multicolumn{3}{c}{$n=30$}\\
method & time & TV & time & TV & tail\\
\midrule
lattice race & 2.5 ms & $1.8\times10^{-4}$ & 4.7 ms & $8.3\times10^{-4}$ & $0.07$\\
Genz, full vector & 46 ms & $1.8\times10^{-4}$ & 1.40 s & $8.3\times10^{-4}$ & $0.07$\\
frequency, matched time & --- & $3.8\times10^{-3}$ & --- & $1.0\times10^{-2}$ & $0.45$, 2 zeros\\
frequency, $10\times$ time & --- & $9.2\times10^{-4}$ & --- & $2.0\times10^{-3}$ & $0.21$\\
Mendell--Elston & 2.0 ms & $1.8\times10^{-3}$ & 18 ms & $3.2\times10^{-2}$ & $1.30$\\
Clark-type & 1.4 ms & $1.1\times10^{-2}$ & 17 ms & $1.6\times10^{-2}$ & $1.33$\\
\bottomrule
\end{tabular}
\end{center}

Genz and the lattice agree to the reference's noise and to each other.
Of the $300\times$ cost ratio at $n=30$, the one-at-a-time interface
supplies an unavoidable $n$-fold component; the remainder reflects the
differing integral representations, tolerances and implementations, and
we do not apportion it further.

At the paper's headline scale the $n$-fold component alone is decisive.
One probability at $n=200$ costs $0.43$ s at \texttt{maxpts}
$=2.5\times10^{5}$, agreeing with the lattice to four or five digits, and
six digits needs multi-second budgets (the committed \texttt{bench.py
genz200}). The full vector therefore costs $86$ s against $26$ ms for the
lattice, a factor of $3\times10^{3}$ before any accuracy matching.

The Mendell--Elston and Clark-type moment recursions are the only
entries whose error grows
silently: its tail entries are off by $e^{1.3}\approx3.7$ with nothing in
the output to say so. The race's tail figure is the reference's own
noise, and lattice self-convergence places its error near $10^{-8}$.

\paragraph{Referee agreement.} A replay harness maps each benchmark
probability to its difference orthant and prices it through the two
field-standard R implementations. \texttt{mvtnorm}'s Genz--Bretz agrees with
the lattice to within its own reported error bound in every family:
$3.8\times10^{-7}$ against a bound of $8.1\times10^{-7}$ at $n=10$, and
$2.7\times10^{-8}$ against $5.8\times10^{-8}$ on a case whose smallest
probability is $3\times10^{-10}$. Botev's minimax tilting agrees to its own
Monte Carlo noise, which in relative terms is $2.2\times10^{-4}$ at
$p\approx3\times10^{-10}$: the tail claim rests on the named accuracy
reference, not only on self-convergence. On the independent family an
adaptive-quadrature referee reaches machine precision and the lattice matches
it to $1.9\times10^{-15}$. An invariance battery (two-runner closed form,
translation, permutation, symmetry) holds to $10^{-16}$ throughout, except
the two-runner comparison at $7\times10^{-9}$, which is the lattice
discretization itself.

\paragraph{Stress boundaries.} An adversarial battery pushes each failure
axis until something gives. Depth: against pure-relative adaptive
quadrature, relative error is $4\times10^{-13}$ at $p\approx10^{-30}$ and
$2\times10^{-8}$ at $p\approx10^{-50}$, with the first genuine breakdown
near $p\approx10^{-119}$, a laggard twenty-five standard deviations behind.
Inversion: the round trip $p\to\mu\to p$ holds $10^{-9}$ in log even when
the target contains a $10^{-10}$ entry or ties at the $10^{-9}$ level. The
Gumbel base reproduces softmax to $10^{-15}$ at $n=1000$; exact duplicates
split exactly, and an $\epsilon$ perturbation moves the split linearly.

The one operative boundary is factor strength. When loadings make implied
correlations exceed roughly $0.9$, the argmax integrand loses smoothness in
factor space and Gauss--Hermite converges slowly in the node count: at
loading scale $4$ the $25$-node rule still sits at $8\times10^{-3}$ total
variation. Scrambled-Sobol factor nodes restore reference-level accuracy
($4\times10^{-4}$) at the same call: heavy correlation wants
quasi-Monte-Carlo nodes, not deeper polynomial rules. The default
escalates the family automatically past sharpness $3$, in both the Python
and R implementations (scrambled Sobol and Halton respectively), which
agree to $5\times10^{-4}$ on the worst case; the figures above are
measured with that escalation disabled, to expose the boundary itself.

\paragraph{Against GHK.} All-$n$ choice probabilities under a rank-2 factor
covariance, against a $2^{20}$-point QMC reference. TV is a bulk metric and
blind to the tails, so the same tail column as above is reported: the
maximum absolute log-error over reference probabilities in
$[10^{-4},10^{-3}]$, which for probabilities this small is the log-odds
error, and which the reference's own counting noise floors near $0.1$.

\begin{center}\small
\setlength{\tabcolsep}{4pt}
\begin{tabular}{rrrrrrrrrr}
\toprule
& \multicolumn{3}{c}{lattice race} & \multicolumn{3}{c}{GHK, $R=10^3$} & \multicolumn{3}{c}{GHK, $R=10^4$}\\
$n$ & time & TV & tail & time & TV & tail & time & TV & tail\\
\midrule
10  & 2.3 ms & $1.8\times10^{-4}$ & --- & 1.0 ms & $8.4\times10^{-3}$ & --- & 8 ms & $2.1\times10^{-3}$ & ---\\
50  & 7.0 ms & $1.7\times10^{-3}$ & $0.11$ & 21 ms & $4.2\times10^{-2}$ & $1.10$ & 214 ms & $1.2\times10^{-2}$ & $0.29$\\
200 & 25 ms  & $2.4\times10^{-3}$ & $0.15$ & 436 ms & $3.2\times10^{-2}$ & $1.50$ & 4.8 s & $1.2\times10^{-2}$ & $0.35$\\
\bottomrule
\end{tabular}
\end{center}

The $n=10$ draw has no reference probability in the tail band. The
lattice's tail figures equal the reference's own noise; GHK's are two to
ten times larger in log terms even at ten times the draws, which is the
tail blindness the TV column cannot show.

Self-convergence places the lattice error near $10^{-8}$. GHK error
contracts as $O(R^{-1/2})$ and its measured complete-vector cost grows
superquadratically, with exponent $2.15$ to $2.8$ across the measured
range. The lattice is deterministic under Gauss--Hermite nodes and
seeded-reproducible under scrambled Sobol, supplies smooth matrix-free
derivatives free of simulation noise, and prices the tails. Simulated
likelihoods can of course be differentiated, but they inherit the noise.

Our claim is bounded as follows. For complete-vector pricing and
inversion under the factor, block, nested and tree covariance grammars,
the shared-field construction removes the per-alternative repetition
intrinsic to standard GHK implementations, and dominates the
per-alternative GHK implementation tested here in both accuracy and cost.
GHK implementations with antithetics, quasi-random sequences or
reordering heuristics would move the constants, not the $n$-fold shape.
Estimation pipelines built around GHK are a larger object than this loop,
and the estimation paragraphs below report where simulation remains
competitive at face value. GHK also retains general rectangle
probabilities and locality-structured covariance, and \citet{botev2017}
remains the reference for single extreme tail probabilities.

\paragraph{Share inversion under sampling noise.} The first estimation
exercise is deliberately the saturated case. Simulate $T=50{,}000$ choices
from one fixed menu of $n=30$ alternatives under a known rank-2 factor
covariance and estimate the mean-zero utilities. With
the menu fixed the model is saturated over the simplex interior, so the
analytic-gradient maximizer is the inversion of the observed shares and the
exercise measures how faithfully each method performs that inversion
under multinomial noise.

Each count is smoothed by one half before any arm sees it, which is the
posterior-mean share vector under the multinomial Jeffreys prior,
$\mathbb{E}[p_i\mid c]=(c_i+\tfrac12)/(T+n/2)$. (It is not the Jeffreys
MAP: the $\mathrm{Dirichlet}(\tfrac12)$ posterior has exponents
$c_i-\tfrac12$, whose mode leaves empty cells on the boundary. Read as a
penalized likelihood in simplex coordinates, add-half is the
$\mathrm{Dirichlet}(\tfrac32)$ mode.) In the saturated model the estimator
is the race inverse of that smoothed vector, and the simulated arms
maximize the same add-half-smoothed criterion, so the comparison stays
like for like.

This is not cosmetic. The unpenalized saturated likelihood has \emph{no}
finite maximizer when any count is zero, which Theorem~\ref{thm:inverse}
states and this design would otherwise walk into: the smallest
alternative here has $p=2.5\times10^{-5}$, so its expected count is
$1.27$ and it is empty with probability $e^{-1.27}=0.28$, so about
eleven of forty replications should see that cell empty. Thirteen of the
forty do, and in this design it is the only cell that ever empties: every
empty-cell replication is that alternative. On those thirteen an
unpenalized optimizer returns a tolerance-dependent point on a boundary
sequence rather than a maximizer. Forty replications, with the arms paired by construction
since every method sees the same draw:

\begin{center}\small
\begin{tabular}{lrrr}
\toprule
estimator & RMSE (median) & RMSE (mean $\pm$ s.e.) & median fit time \\
\midrule
analytic-gradient smoothed inversion & $0.0168$ & $0.0168\pm0.0006$ & $2.4$ s \\
MSL, $R=100$ & $0.0286$ & $0.0280\pm0.0006$ & $3.5$ s \\
MSL, $R=1000$ & $0.0184$ & $0.0185\pm0.0005$ & $39.9$ s \\
\bottomrule
\end{tabular}
\end{center}

\noindent The exact path is both faster and more accurate than the
$R=100$ simulator here, and beats the $R=1000$ accuracy in a sixteenth
of its time. Because the arms share every draw, the paired differences
are the sharper statement: against $R=100$ the exact arm is better by
$0.0112\pm0.0006$, nearly nineteen standard errors, winning $39$ of $40$
replications; against $R=1000$ by $0.0017\pm0.0002$, seven standard
errors, winning $33$ of $40$.

The $R=1000$ margin is small in absolute terms, which is the honest
characterization: enough simulation draws do approach the exact answer,
at forty times the cost per fit, and the surviving paired gap is
consistent with finite-$R$ simulation bias rather than multinomial
sampling noise, though pairing alone cannot rule out another systematic
implementation effect. The timing column is from
an unloaded machine at the original eight replications; the accuracy
columns are from all forty and are insensitive to load.

\paragraph{Covariate estimation.} The overidentified case behaves
differently, and we report it as we found it. Draw observation-specific
design matrices $X_t$ ($n=10$ alternatives, $d=3$ covariates,
$T=800$ observations), utilities $\mu_t=X_t\beta$, one choice per
observation, covariance known, and estimate $\beta$ by individual-level
maximum likelihood: the exact path uses lattice probabilities with the
analytic score (one Jacobian row per observation, a single field pass);
the simulated path uses this paper's own Rust GHK with common random
numbers and finite-difference gradients. Eight replications, means with
standard errors:

\begin{center}\small
\begin{tabular}{lrr}
\toprule
estimator & RMSE$(\hat\beta)$ & median fit time \\
\midrule
analytic-gradient MLE & $0.0363\pm0.0066$ & $41.3$ s \\
MSL, $R=100$ & $0.0380\pm0.0073$ & $5.4$ s \\
MSL, $R=1000$ & $0.0367\pm0.0067$ & $39.6$ s \\
\bottomrule
\end{tabular}
\end{center}

\noindent All three estimators reach the statistical floor. With three
parameters and eight hundred observations the simulation error averages
out, and a well-implemented simulator at $R=100$ is the fastest of the
three at this size. The exact path's time is dominated by per-observation
Python overhead: the compiled kernel accounts for under a second of the
$41$, so the comparison at $n=10$ measures plumbing, not method.

What the exact likelihood buys in estimation is not this row of the
table. It is freedom from the choice of $R$, gradients without simulation
noise, the cost law as $n$ grows, and diagnostics that simulation smooths
over.

That last point is not hypothetical. On the classic Fishing data, the
unrestricted-covariance MNP likelihood evaluated exactly is
boundary-seeking, still rising at loading norms near $4000$: a ridge
that the finite-$R$ simulated objective had obscured. The behavior is
reproducible from the committed estimator (\texttt{winning.mnprobit}),
which detects the ridge and reports it as a \texttt{boundary\_}
diagnostic rather than returning the interior point simulation would have
manufactured.

For interpretability of that norm: the parameterization is rank-2
loadings with the reference alternative's rows fixed at zero and unit
idiosyncratic variances, so global scale is pinned by the idiosyncratic
normalization and the growing norm is a genuine choice-relevant
direction, not the scale gauge. A profile likelihood along that
normalized direction is the diagnostic a full treatment would add.

\paragraph{The Genz--Bretz score.} Our quadrature comparator used
\texttt{pmvnorm}, which exposes no score, with finite-difference gradients
($31$ vector evaluations per step at $n=30$, roughly eight times the entire
exact-gradient share-inversion fit). The \texttt{mvtnorm} manual (v1.4-2)
also documents \texttt{slpmvnorm} and, for the reduced-rank covariance
$BB'+\mathrm{diag}(D)$ of this paper's benchmarks, \texttt{lpRR} and
\texttt{slpRR}: simulated log-likelihoods and score functions that
integrate over the $K$ factor dimensions. That interface is itself the
conditioning construction of \S\ref{sec:factor} with Monte Carlo in place
of quadrature and no shared field across alternatives, which makes it the
natural structure-aware comparator for a fuller estimation study.

\paragraph{GHK protocol.} The GHK figures throughout use this paper's own
Rust implementation: per-alternative sequential conditioning on the
reference-differenced covariance with a fresh Cholesky factorization per
alternative, pseudorandom draws (xoshiro family) with a per-alternative
seed offset so that parameter evaluations share common random numbers, no
antithetics, no quasi-random sequence, no variable-reordering heuristic,
parallelized across alternatives. Cost figures are complete-vector; the
large-$n$ columns of the introduction are extrapolations of the measured
law and are displayed as such.

\paragraph{Scale.} Ten million contestants price in $245$ seconds under the
block grammar (ten thousand clusters, $257$-point lattice, the streaming
field pass; the committed \texttt{bench.py tenmillion}), and the Rust
classic-calibration kernel runs $232$ times the speed of the pure-Python
implementation. Throughput at this size is a systems statement; the accuracy
evidence lives at the scales the referees above can reach.

\paragraph{Parity.} One vector file embeds the inputs and outputs of
$22$ scenarios spanning the principal implemented grammars, boundary
cases and numerical identities used above. The four language
implementations of \S\ref{sec:avail} replay them: nineteen to within
$10^{-7}$ (most to $10^{-10}$), and the three fit- or optimizer-mediated
scenarios to between $5\times10^{-4}$ and $5\times10^{-3}$.

\section{Summary}\label{sec:summary}

The pieces of this paper are individually old. The one-dimensional
representation of the choice event is in \citet{domencich1975}, who
judged it computationally intractable for correlated probit. The convex
duality of calibration is in \citet{li2018}. The tie-density
derivatives and the surplus function's Laplacian structure are in
\citet{muller2022}, albeit independently arrived at from a different direction. A Hessian vector product at the cost of a gradient
is \citet{pearlmutter1994}. Building the survival field once and dividing each contestant out to collapse the forward vector to $O(nL)$ is performed here differently than \citealt{cotton2021}, but that is where the idea is borrowed and crucially, this is the point of departure from other smooth methods that require iteration over all items.

Herein we have organized the pieces into an efficient algorithm where a
single shared survival field renders probabilities for all $n$
alternatives at once and supplies the own-coordinate slopes that
precondition the million-alternative calibrations reported here.

We have applied and interpreted the dense tie-density Jacobian: a weighted graph Laplacian and a probability flux through the boundary between winning regions computable in linear work per lattice point. We have supplied a covariance grammar that, subject to its self-described limitations, lifts probit models and their cousins out of the small-to-medium scale domain and makes them practical for a variety of modern problems where previously, only logit was deemed to be a starter.

\section{Availability}\label{sec:avail}
\texttt{pip install winning} is pure Python; \texttt{winning[fast]} adds the
Rust kernels as one abi3 wheel per platform. A base-R package mirrors the full
interface, and a zero-dependency browser port drives a live demonstration at
\texttt{winning.microprediction.org}.

\appendix
\section{Extrapolated GHK cost brackets}\label{app:extrap}

The lattice column below is measured; the GHK column is an
extrapolation of the two measured exponents ($2.15$ and $2.8$) far
beyond any regime in which GHK has been run, and cache effects,
parallelism, Cholesky reuse or a change of regime would move it.

\begin{center}\small
\begin{tabular}{rrrr}
\toprule
$n$ & lattice, measured & GHK at $10^4$ draws, cost-law bracket & multiplier \\
\midrule
$10^4$ & $0.18$ s & $20$ hours -- $3.7$ days & $4\times10^{5}$ -- $2\times10^{6}$ \\
$10^5$ & $2.7$ s & $0.3$ -- $6.5$ years & $4\times10^{6}$ -- $8\times10^{7}$ \\
$10^6$ & $29$ s & $46$ -- $4{,}100$ years & $5\times10^{7}$ -- $4\times10^{9}$ \\
\bottomrule
\end{tabular}
\end{center}

Nobody prices a million-alternative field with GHK; the brackets say
what the measured law implies about why. The ten-million forward-field
benchmark of \S\ref{sec:num}, which is actually run, is the meaningful
systems demonstration.

\end{document}